\documentclass[conference,10pt,a4paper]{IEEEtran}

\usepackage[colorlinks,linkcolor={blue},citecolor={blue},urlcolor={red},breaklinks=true,final]{hyperref}
\usepackage{csquotes}

\usepackage{balance}

\newcommand{\contprob}[3]{P^{#1}_{#2}({#3})}
\newcommand{\continf}[2]{\contprob{\omega}{#1}{#2}}
\newcommand{\sigmaact}{\sigma_{\mathsf{act}}}
\newcommand{\sigmapass}{\sigma_{\mathsf{pass}}}

\newcommand{\taupass}{\tau_{\mathsf{pass}}}
\newcommand{\ofrac}[2]{#1}

\newcommand{\rhoact}{\rho_{\mathsf{act}}}

\newcommand{\rhopass}{\rho_{\mathsf{pass}}}
\newcommand{\rhoactof}[1]{\rho^{#1}_{\mathsf{act}}}

\newcommand{\rhopassof}[1]{\rho^{#1}_{\mathsf{pass}}}
\newcommand{\rhofailof}[1]{\rho^{#1}_{\mathsf{fail}}}

\newcommand{\citet}[1]{\cite{#1}}
\newcommand{\noqed}{\def\qed{}}                             %
\newcommand{\ml}{\cdot}

\newcommand{\plushalf}{+_{\frac{1}{2}}}

\newcommand{\mone}{{\text{\kern.5pt\rmfamily-}\mathsf{\kern-.5pt1}}}

\usepackage{graphicx}
\usepackage[nosumlimits,nonamelimits]{amsmath}
\usepackage{booktabs}   %
\usepackage{subcaption} %

\usepackage{array}
\usepackage{mathtools}
\usepackage{enumitem}

\usepackage{amsthm}
\usepackage{proof}
\usepackage{xfrac}

\usepackage{tikz-cd}

\tikzset{
      commutative diagrams/.cd
    , arrow style=tikz
    , diagrams={>=stealth}
    , row sep=large
    , column sep = huge
}
 
\tikzset{
    cong/.style={draw=none,edge node={node [sloped, allow upside down, auto=false]{$\cong$}}}
  , iso/.style={draw=none,every to/.append style={edge node={node [sloped, allow upside down, auto=false]{$\cong$}}}}
}

\tikzset{shiftarr/.style={
        rounded corners,%
        to path={--([#1]\tikztostart.center)
                     -- ([#1]\tikztotarget.center) \tikztonodes
                     -- (\tikztotarget)},
}}

\tikzset{cong/.style={draw=none,edge node={node [sloped, allow upside down, auto=false]{$\cong$}}},
         Isom/.style={draw=none,every to/.append style={edge node={node [sloped, allow upside down, auto=false]{$\cong$}}}}}

\usepackage{etoolbox} %

\usepackage{todos}
\usepackage{etoolbox} %

\usepackage{ifdraft} 

\makeatletter
\@draftfalse
\makeatother

\ifdraft{
	\usepackage[displaymath,mathlines,switch]{lineno}
	\linenumbers

	\newcommand*\linenomathpatch[1]{%
		\cspreto{#1}{\linenomath}%
		\cspreto{#1*}{\linenomath}%
		\csappto{end#1}{\endlinenomath}%
		\csappto{end#1*}{\endlinenomath}%
	}

	\linenomathpatch{equation}
	\linenomathpatch{gather}
	\linenomathpatch{multline}
	\linenomathpatch{align}
	\linenomathpatch{alignat}
	\linenomathpatch{flalign}

	\usepackage{showlabels} 
	
	\usepackage[layout=footnote,draft]{fixme}

}{
 \usepackage[layout=footnote,final]{fixme}
}

\FXRegisterAuthor{sg}{asg}{SG}	%
\FXRegisterAuthor{pl}{apl}{PL}	%
\FXRegisterAuthor{nb}{anb}{NB}	%

\newenvironment{spaceout}[1]{\begin{displaymath}\setlength{\extrarowheight}{3pt}\begin{array}{#1}}{\end{array}\setlength{\extrarowheight}{0pt}\end{displaymath} \noindent}
\newcommand{\tuple}[1]{\langle{#1}\rangle}
\newcommand{\failof}[1]{{#1}_{\mathsf{fail}}}
\newcommand{\tesum}{\textstyle{\sum}}

\newcommand{\catr}{\ensuremath{\mathcal{R}}}

\makeatletter
\newcommand{\Rrightup}[1]{{\mathrel{\;\overset{\raisebox{3.2pt}{$\scriptstyle #1$\;}}{\smash{\xRightarrow{\hspace{1.75em}}}}\;}}}
\makeatother

\newcommand{\vict}{\mathsf{Vict}}
\newcommand{\finfound}{\mathsf{FinFound}}
\newcommand{\namesig}{S}
\newcommand{\hsig}{H_{\namesig}}
\newcommand{\Ops}{K}

\newcommand{\cstc}{{\mathsf{c}}}
\newcommand{\cona}{\mathsf{a}}
\newcommand{\conb}{\mathsf{b}}
\newcommand{\conc}{\mathsf{c}}
\newcommand{\tensor}{\,\otimes\,}

\newcommand{\eqdef}{\stackrel{\mbox{\rm {\tiny def}}}{=}}
\newcommand{\setbr}[1]{\{{#1}\}}

\newcommand{\bnfgo}{\Coloneqq}

\newcommand{\nats}{\ensuremath{\mathbb{N}}}
\newcommand{\betwixt }{\hspace{1em}}

\newcommand{\set}{\mathbf{Set}}
\newcommand{\pset}{\mathcal{P}}

\newcommand{\cate}{\mathcal{E}}

\newcommand{\prog}{\mathsf{Prog}}

\newcommand{\twodots}{\mathinner {\ldotp \ldotp}}

\DeclareFontFamily{U}{mathc}{}
\DeclareFontShape{U}{mathc}{m}{it}{<->s*[.93] mathc10}{}
\DeclareMathAlphabet{\morph}{U}{mathc}{m}{it}

\newcommand{\inpu}[1]{\mathsf{Req}\,{#1}?} %

\newcommand{\sfinpu}[1]{\morph{Req}\,{#1}?}
\newcommand{\tracesof}[1]{\morph{Traces}\,{\left(#1\right)}}

\newcommand{\smin}{\!\in\!}
\newcommand{\supp}[1]{\mathsf{supp}({#1})}

\newcommand{\arity}[1]{\mathsf{Ar}({#1})}
\newcommand{\inarity}[1]{\arity{#1}}

\usepackage[nosumlimits,nonamelimits]{amsmath}
\usepackage{amsthm,amssymb}

\usepackage[capitalize,noabbrev,nameinlink]{cleveref} %

\theoremstyle{plain}

\newtheorem{theorem}{Theorem}[section]

\newtheorem{proposition}[theorem]{Proposition}
\newtheorem{lemma}[theorem]{Lemma}
\newtheorem{corollary}[theorem]{Corollary}

\theoremstyle{definition}

\newtheorem{definition}[theorem]{Definition}
\newtheorem{example}[theorem]{Example}

\newcommand{\dist}{\mathcal{D}}
\newcommand{\distfin}{\mathcal{D}_{\mathsf{f}}}

\newcommand{\Nat}{{\mathbb N}}

\newcommand{\comma}{,\operatorname{}\linebreak[1]}

\usepackage{microtype}
\usepackage{cite}
\usepackage{amsmath,amssymb,amsfonts}
\usepackage{algorithmic}
\usepackage{graphicx}
\usepackage{textcomp}
\usepackage{xcolor}
\def\BibTeX{{\rm B\kern-.05em{\sc i\kern-.025em b}\kern-.08em
    T\kern-.1667em\lower.7ex\hbox{E}\kern-.125emX}}

\usepackage[english]{babel}
\addto{\captionsenglish}{}

\allowdisplaybreaks
\begin{document}

\title{Probabilistic Strategies: Definability and\\
the Tensor Completeness Problem
}

\author{\IEEEauthorblockN{Nathan Bowler}
\IEEEauthorblockA{
	\textit{Department of Mathematics} \\
	\textit{University of Hamburg}\\
	Hamburg, Germany \\
	nathan.bowler@uni-hamburg.de
}
\and
\IEEEauthorblockN{Sergey Goncharov}
\IEEEauthorblockA{
	\textit{School of Computer Science} \\
	\textit{University of Birmingham}\\
	Birmingham, UK \\
	s.goncharov@bham.ac.uk
}
\and
\IEEEauthorblockN{Paul Blain Levy}
\IEEEauthorblockA{
	\textit{School of Computer Science} \\
	\textit{University of Birmingham}\\
	Birmingham, UK \\
	p.b.levy@bham.ac.uk
}
}

\maketitle

\begin{abstract}
Programs that combine I/O and countable probabilistic choice, modulo either
bisimilarity or trace equivalence, can be seen as describing a probabilistic
strategy. For well-founded programs, we might expect to axiomatize bisimilarity
via a sum of equational theories and trace equivalence via a tensor of
such theories. This is by analogy with similar results for nondeterminism,
established previously.
While bisimilarity is indeed axiomatized via a sum of theories, and the tensor
is indeed at least sound for trace equivalence, completeness in general,
remains an open problem. Nevertheless, we show completeness in the case that
either the probabilistic choice or the I/O operations used are finitary. We
also show completeness up to impersonation, i.e.\ that the tensor theory regards
trace equivalent programs as solving the same system of equations. This entails 
completeness up to the cancellation law of the probabilistic choice operator.

Furthermore, we show that a probabilistic trace strategy %
arises as the semantics of a well-founded program iff it is victorious. This
means that, when the strategy is played against any partial counterstrategy, the
probability of play continuing forever is zero.

We link our results (and open problem) to particular monads that can be used to
model computational effects.
\end{abstract}


\section{Introduction}

\label{sec:intro}

\subsection{The axiomatization problem} \label{sec:axiom}

{Program equivalence}, and the question of how to axiomatize it, is a classical subject in theoretical
computer science.  Sometimes, the programs studied are higher-order or
concurrent. In this paper, however, we deal with first-order, sequential programs that
perform I/O operations and make discrete probabilistic choices.
 Although this may seem a limited setting, the problem of axiomatizing \emph{trace equivalence} (a natural notion
of equivalence) has turned out to be surprisingly tough.

\paragraph*{Determinism}  To begin, we consider an \emph{illustrative} deterministic language with a
nullary operation ($\mathsf{Bye}$), a binary operation
($\mathsf{Happy}?$) and an infinitary operation ($\mathsf{Age}?$).
The syntax is as follows:
\begin{equation}\label{eq:sim-lang}
\begin{aligned}
~~ M,N \Coloneqq  &\;\mathsf{Age}?(M_i)_{i\in\Nat} %
  						\mid\mathsf{Happy}?(M,N) %
  						\mid\mathsf{Bye}. %
\end{aligned}
\end{equation}
The programs of the above language describe I/O interaction with the environment: 
$\mathsf{Age}?(M_i)_{i\in\Nat}$ asks the user for their age, awaits a natural 
number input $i$ and proceeds as~$M_i$ accordingly; 
$\mathsf{Happy}?(M,N)$ asks if the user is happy, awaits a boolean answer and 
proceeds as $M$ in the positive case and as $N$ in the negative case; 
$\mathsf{Bye}$ outputs ``Bye'' and finishes the dialogue.
Formally, the semantics is given via a suitable transition system. 
We emphasize that the syntax, though infinitary, is still \emph{well-founded},
as it is inductively defined. For example, there is no program $M =
\mathsf{Happy}?(\mathsf{Bye},M)$.

Now we turn to a more general framework.  Recall that a
\emph{signature} consists of a set $\Ops$ of \emph{operation symbols}, each equipped with an \emph{arity}, which is a (possibly infinite) set.  The arity of $k$ is written $\arity{k}$.  
  Thus a signature $\namesig$ is an indexed family of sets $(\arity{k})_{k\in\Ops}$.  
  
In the context of this paper, we may call $k \smin
\Ops$ 
an \emph{output} and $i \smin \arity{k}$ an \emph{input} following $k$.  In the previous example, the three outputs are  $\mathsf{Bye}$, $\mathsf{Happy}$ and
$\mathsf{Age}$.  There are no inputs following $\mathsf{Bye}$, an input following $\mathsf{Happy}$ is a boolean, and one following $\mathsf{Age}$ is a natural number.
Each signature $\namesig$, such as the one above, gives rise to a language with the following syntax:
\begin{equation}
  \label{eq:sim-lang-gen}
  \begin{aligned}
    ~~ M \Coloneqq  &\;\mathsf{Req}\,k?(M_i)_{i\in\arity{k}}
\end{aligned}
\end{equation}
Here, $\mathsf{Req}$ stands for ``Request''. The program $\mathsf{Req}\,k?(M_i)_{i\in\arity{k}}$
 outputs $k$, awaits an input $i\smin\arity{k}$ and proceeds as $M_i$.  
 In our illustrative language, we wrote $\mathsf{Req}\,\mathsf{Bye}?$
 as $\mathsf{Bye}$, wrote 
$\mathsf{Req}\,\mathsf{Happy}?$ as $\mathsf{Happy}?$ and wrote 
$\mathsf{Req}\,\mathsf{Age}?$ as $\mathsf{Age}?$, for the sake of readability.

Borrowing terminology from game semantics, a \emph{play} is  a sequence of actions,
alternating between output and input.   Each program has a
\emph{trace set}---the set of all plays that might arise when the program
is run.  For example, traces of the program
\begin{align*}
\mathsf{Age}?(\mathsf{Bye},\mathsf{Bye},\mathsf{Bye},\mathsf{Bye},\mathsf{Bye}, \mathsf{Happy}?(\mathsf{Bye},\mathsf{Bye}),\mathsf{Bye},\ldots ) 
\end{align*}
include $\mathsf{Age}?0.\mathsf{Bye}$, $\mathsf{Age}?5.\mathsf{Happy}?$, %
$\mathsf{Age}?5.\mathsf{Happy}?1.\mathsf{Bye}$. Two programs are \emph{trace equivalent}
when they have the same trace set. The language is semantically
straightforward: two programs are 
trace equivalent iff they are bisimilar, and iff they are
syntactically equal.

\begin{figure*}[t]
\centering
 \begin{subfigure}[b]{0.4\textwidth}
     \centering
  \begin{flalign*}
  &&M +_p M 			 &\, = M& \\*
  &&M +_p N				 &\, = N +_{1-p} M \\
  &&(M +_p N)+_q P &\, = M +_{\frac{p}{p+q-pq}} (N+_{p+q-pq} P) 
\end{flalign*}
  \caption{The convex laws~\cite{Stone49}.}
 \end{subfigure}
 \hspace{1.5em}
 \begin{subfigure}[b]{0.4\textwidth}
     \centering
\begin{align*}
  \sum_{n\in\nats} \delta_{n,m}\ml M_n & = M_m \\
  \sum_{n\in\nats} p_n\ml \sum_{m\in\nats} q_{n,m} M_m & = 
  \sum_{m\in\nats}\Bigl(\sum_{n\in\nats}p_n\ml q_{n,m}\Bigr)\ml M_m
\end{align*}
  \caption{The $\omega$-convex laws~\cite{Konig86}.}
 \end{subfigure}
    \caption{Axioms of probabilistic choice operators.}
    \label{fig:prob-laws}
\end{figure*}

\paragraph*{Adding probabilistic choice}

We next consider two %
probabilistic
extensions (binary and countable):
\begin{align}\label{eq:add_prob}
                                       M,N \Coloneqq  &\;
                                                        \inpu{k}(M_i)_{i
                                                       \in\arity{k}}
                                                        \mid M+_p N  \\
  M,N \Coloneqq & \; \inpu{k}(M_i)_{i\in\arity{k}}\mid
                  \sum_{n\in\nats} p_n\ml M_n
\end{align}
The first extension provides programs of the form \mbox{$M+_p N$}, where $p$
ranges over the interval $(0,1)$.  Such a program chooses either 0 or 1 with
probability $p$ and $1-p$ respectively, then executes $M$ or $N$
accordingly.

The second extension provides programs of the form ${\sum_{n\in\nats}
p_n\ml M_n}$, where $(p_n)_{n\in\nats}$ is a sequence of nonnegative
reals that sum to 1.  Such a program chooses $n\in\nats$ with
probability $p_n$, then executes $M_n$.

These calculi are the probabilistic counterparts of nondeterministic
calculi studied in~\cite{BowlerLevyEtAl18}.

\paragraph*{Standard equations for probabilistic choice}

Our project to axiomatize equivalence of probabilistic program begins with 
standard laws from the literature. For binary choice, we take the equations 
displayed in~\cref{fig:prob-laws}(a) where~$\delta_{n,m}$ is defined to be $1$ if $n=m$, and 0 otherwise. 
For countable choice, we take the equations displayed in~\cref{fig:prob-laws}(b).  The latter is also
suitable for binary choice: we simply restrict \mbox{$\sum_{n\in\nats} p_n
\ml M_n$} to the case that all but finitely many $p_n$ are zero.
These laws are only
concerned with probabilistic choice, and say nothing specifically
about the I/O operations.  In each case, we shall see that
they are sound and
complete for bisimulation.  In other words, bisimilarity (the largest
bisimulation) is the least congruence
that includes those laws.  This is in line with
results established for nondeterministic languages~\cite{BowlerLevyEtAl18}.

\paragraph*{Trace semantics}

In each language, a program $M$ has a \emph{trace semantics}, which is a
real-valued function on plays, of a certain kind.  The idea is that this function 
sends each play $s$ to the probability that $M$ generates the output
that $s$
describes if it receives the input that $s$ describes. Consider an example:
\begin{align*}
\mathsf{Happy?}(\mathsf{Bye},\mathsf{Bye}+_{\sfrac{1}{2}}\mathsf{Happy?}(\mathsf{Bye}+_{\sfrac{1}{3}}\mathsf{Happy?}(\mathsf{Bye},\mathsf{Bye}))).
\end{align*}
The trace semantics can be described as follows: it sends $\mathsf{Happy?}$
to $1$, $\mathsf{Happy?}0.\mathsf{Bye}$ to $1$, $\mathsf{Happy?}1.\mathsf{Bye}$ 
and $\mathsf{Happy?}1.\mathsf{Happy?}$ to $\sfrac{1}{2}$,
$\mathsf{Happy?}1.\mathsf{Happy?}0.\mathsf{Bye}$, $\mathsf{Happy?}1.\mathsf{Happy?}1.\mathsf{Happy?}$, 
$\mathsf{Happy?}1.\mathsf{Happy?}1.\mathsf{Happy?}0.\mathsf{Bye}$, and $\mathsf{Happy?}1.\mathsf{Happy?}1.\mathsf{Happy?}1.\mathsf{Bye}$ 
to $\sfrac{1}{6}$ -- the remaining plays are sent to $0$.

Again, the corresponding story for nondeterministic languages was
presented in~\cite{BowlerLevyEtAl18}.  There, the trace semantics of a 
program is simply a set of plays, not a function on plays.

\paragraph*{Axiomatizing trace equivalence}

Programs with the same trace semantics are \emph{trace equivalent}.
This is congruence relation coarser than bisimilarity.  How shall we axiomatize it?

For the finitely probabilistic language, we extend the convex laws
with laws saying that I/O operations \emph{commute} with~$+_p$, in the sense of the top equation in \cref{fig:tens-laws}.
For example
\begin{align*}
\mathsf{Happy?}(M,N) +_p \mathsf{Happy?}(M',N') & = \\* \mathsf{Happy?}& (M+_p M', N+_p N')
  \\
\mathsf{Age?}(M_i)_{i \in\Nat} +_p  \mathsf{Age?}(M'_i)_{i\in\Nat}
                            & = \mathsf{Age?}(M_i +_p M'_i)_{i\in\Nat}
\end{align*}%
The least congruence that includes all these laws is called
\emph{tensor-equivalence}.  

Likewise, for the countably probabilistic language, we extend the
$\omega$-convex laws with laws saying that each I/O operation
\emph{commutes} with $\sum_{n\in\nats}$, in the sense of the bottom equation in \cref{fig:tens-laws}.
For example
\begin{align*}
\sum_{n\in\nats} p_n \ml \mathsf{Happy?}(M_n,N_n) 
                            &\, =\\ \mathsf{Happy?} \Bigl(\sum_{n \in
                                  \nats} &p_n\ml M_n, \sum_{n\in\nats}
                                  p_n \ml N_n\Bigr)
\\
\sum_{n\in\nats}p_n\ml \mathsf{Age?}(M_{n,i})_{i \in\Nat} 
                            &\, = \mathsf{Age?}\Bigl(\sum_{n\in\nats}p_n\ml  M_{n,i}\Bigr)_{i\in\Nat}
\end{align*}%
Again, the least congruence that includes all these laws is called
\emph{tensor-equivalence}.

It is straightforward to see that tensor-equivalence is sound for
trace semantics.  In other words, trace equivalence implies tensor
equivalence.  But is it complete?

In the nondeterministic language, whether the nondeterministic choice
is binary or countable, the answer is yes~\cite{BowlerLevyEtAl18}.  So we might expect
a similar result for probability.  

In fact, we shall prove that tensor-equivalence is complete in two
cases:
\begin{itemize}
\item when only $+_p$ is used, but not $\sum_{n\in\nats}$;
\item when the signature is finitary, i.e.\ $\arity{k}$ is finite for
  each~$k \smin \Ops$.
\end{itemize}
To summarize, we  establish tensor completeness in the case that \emph{either} the probabilistic choice
\emph{or} the I/O signature is finitary.  If both are infinitary, completeness is an open
question.
Nevertheless, even in this case, we prove a weak form of completeness
called \emph{completeness up to impersonation}. This means that for trace
equivalent programs~$M$ and~$N$, we can---up to tensor-equivalence---see $M$
as the ``canonical'' solution to a system of equations that also has~$N$ as a
solution. Moreoever, there is a program $P$ such that $M \plushalf P$ and $N
\plushalf P$ are tensor-equivalent.

\subsection{Completeness problem by example}\label{sec:exa}
In this section, we allude to our informal introduction of trace equivalence and 
tensor-equivalence,
to give some concrete examples of the completeness problem -- the relevant formal 
details can be found in \Cref{sec:prelims,sec:sys_trace}.

\begin{figure*}
\begin{align*}
\inpu{k} (M_i)_{i\in\arity{k}} +_p  \inpu{k}
  (M'_i)_{i\in\arity{k}}  =&\, \inpu{k} (M_i +_pM'_i)_{i
                                 \in\arity{k}} \\[1ex]
\sum_{n\in\nats} p_n \ml \inpu{k} (M_{n,i})_{i \in
   \arity{k}}  =&\,  \inpu{k} (\sum_{n\in\nats} p_n \ml  M_{n,i})_{i \in
                     \arity{k}}
\end{align*}%

   \caption{Tensor laws for $+_p$ and $\sum_{n\in\nats}p_n$.}
    \label{fig:tens-laws}
\end{figure*}

We will stick to the signature consisting of one binary operation~$\star$, a countable 
set of constants $C$ with no equations. We will write countable sums $\sum_{n\in\Nat} p_n\ml M_n$
as series $p_0\ml M_0+p_1\ml M_1+\ldots$ By writing $M\equiv N$ we mean that the 
equivalence of $M$ and $N$ is provable in the tensor logic. 
\begin{example}\label{exa:trace1}
Let $\cona,\conb,\conc\in C$. Then it can be checked that the following programs 
are trace-equivalent:
\begin{align*}
   \quad M \eqdef\, & \tfrac{1}{2}\cdot\cona \star\conc + \tfrac{1}{4}\cdot\conb \star (\cona \star \conc) + \tfrac{1}{8}\cdot \conb \star (\conb \star (\cona \star \conc))\\*
    & + \tfrac{1}{16}\cdot \conb \star (\conb \star (\conb \star (\cona \star \conc))) + \cdots \\[1ex]
       N \eqdef\, & \tfrac{1}{2}\cdot\conb \star \conc + \tfrac{1}{4}\cdot \cona \star (\conb \star \conc) + \tfrac{1}{8}\cdot \cona \star (\cona \star (\conb \star \conc))\\*
    & + \tfrac{1}{16}\cdot \cona \star (\cona \star (\cona \star (\conb \star \conc))) + \cdots
\end{align*}
How do we prove that $M$ and $N$ are equal in the tensor theory?
Rewriting $M$ by the second equation in~\cref{fig:tens-laws} and subsequently applying 
arithmetic simplifications gives $(\tfrac{1}{2}\cdot\cona + \tfrac{1}{2}\cdot\conb) \star ({\tfrac{1}{2}\cdot\conc} + \tfrac{1}{2}\cdot M)$, which does
 not help. However, a more ingenious approach works out:
\begin{align*}
   M &\,\equiv\tfrac{1}{4}\cdot\cona \star \conc + \tfrac{2}{8}\cdot\bigl(\tfrac{1}{2}\cdot \cona \star \conc + \tfrac{1}{2}\cdot\conb \star (\cona \star \conc)\bigr)\\
         &\quad + \tfrac{3}{16}\cdot\bigl(\tfrac{1}{3}\cdot\cona \star \conc + \tfrac{1}{3}\cdot\conb \star (\cona \star \conc) + \tfrac{1}{3}\cdot\conb \star (\conb \star (\cona \star \conc))\bigr)\\[1ex]
         &\quad + \cdots \\
     &\,\equiv \tfrac{1}{4}\cdot\cona \star\conc + \tfrac{2}{8}\cdot\bigl(\tfrac{1}{2} \cona \star (\cona \star \conc) + \tfrac{1}{2} \conb \star \conc\bigr)\\
         &\quad + \tfrac{3}{16}\cdot\bigl(\tfrac{1}{3}\cdot\cona \star (\cona \star (\cona \star \conc)) + \tfrac{1}{3}\cdot\conb \star \conc + \tfrac{1}{3}\cdot\conb \star (\conb \star \conc)\bigr)\\[1ex]
         &\quad + \cdots \\
     & \quad\betwixt \text{\{by completeness for finitely probabilistic programs\}} \\
     &\,\equiv \tfrac{1}{4}\cdot\cona \star \conc + \tfrac{1}{8}\cdot\cona \star (\cona \star \conc) + \tfrac{1}{16}\cdot\cona \star (\cona \star (\cona \star \conc)) + \cdots \\
     &\quad + \tfrac{1}{4}\cdot\conb \star \conc + \tfrac{1}{8}\cdot\conb \star (\conb \star \conc) + \tfrac{1}{16}\cdot\conb \star (\conb \star (\conb \star \conc)) + \cdots \\
     & \quad\betwixt \text{\{by symmetric reasoning\}}\\
   &\,\equiv N. 
\end{align*}
This also follows from our general result for finite signatures (\Cref{thm:comp_fin}).
\end{example}

\begin{example}\label{exa:trace2}
In the same signature %
fix a family of constants $(\cstc_{n}^{k}\in C)_{n\in\nats,k\leq n}$, and a family
of reals $(p_n>0)_{n\in\nats}$, such that $\sum_{n\in\nats} p_n = 1$. Define the 
following programs:
\begin{align*}
	Q_0 =&\, \sum_{n\in\nats} \frac{p_{n+1}}{1- p_0}\cdot \cstc_{n+1}^{0}\\*
	Q_1 =&\, \sum_{n\in\nats} \frac{p_{n+1}}{1- p_0 - p_1}\cdot \cstc_{n+1}^{1}\\
	     &\,\vdots\\
	Q_k =&\, \sum_{n\in\nats} \frac{p_{n+k+1}}{1- p_0-\ldots -p_{k}}\cdot \cstc_{n+k+1}^{k}\\
	&\,\vdots
\end{align*} 
It can be checked that the following terms are trace equivalent
\begin{align*}
  M\eqdef\; & p_0\ml\cstc_{0}^{0} + p_1\ml \cstc_{1}^{1}\star\cstc^{0}_{1}+ p_2\ml(\cstc_{2}^{2}\star\cstc^{1}_{2})\star\cstc^{0}_{2}+\ldots,\\*
  N\eqdef\; & p_0\ml\cstc_{0}^{0} + p_1\ml \cstc_{1}^{1}\star Q_0+ p_2\ml(\cstc_{2}^{2}\star Q_1)\star Q_0+\ldots.
\end{align*}
Indeed, in both cases, the relevant traces are mapped to probabilities as follows: $c_0^0\mapsto p_0$,
$\star.0.c_1^1\mapsto p_1$, $\star.1.c_0^1\mapsto p_1 = (1-p_0)\cdot\frac{p_1}{1-p_0} = p_1\cdot\frac{p_1}{1-p_0} + p_2\cdot\frac{p_1}{1-p_0} + \ldots$, etc.
Our general~\Cref{thm:comp_fin} entails that~$M\equiv N$.
\end{example}

\subsection{The definability problem}

As we have seen, each program has a trace semantics, which is
real-valued function on plays.   Such a function
 satisfies some basic conditions that makes it a ``strategy''  (again
 borrowing terminology from game semantics). 
 
Although this would be the case even for non-well-founded programs,
our programs are in fact \emph{well-founded}.  This leads to a
definability question.  Which strategies are the trace
semantics of a program in our language?  %
We answer this question in \cref{sec:victor}.  It turns out that a strategy is definable  by a program if and only if, when played against any 
``partial counterstrategy'' (a deterministic way of providing input
that may fail), the probability of continuing forever is zero.  We
call such a strategy \emph{victorious}.

The corresponding question for the nondeterministic languages was
answered in~\cite{BowlerLevyEtAl18}, using the notion of  a ``{well-foundedly total}'' strategy.

\subsection{Sum and tensor of monads}

Following~\citet{Moggi91,PlotkinPower02},  %
we obtain a more abstract view of these results, by modelling each
computational effect or 
combination of effects with a suitable monad on $\set$.
Furthermore, we can speak of the sum or tensor of theories, which
generate the sum or tensor of the respectively generated monads.\footnote{For
   monads that do not arise from a theory, the sum and the tensor do not
   always exist~\cite{HylandLevyEtAl07,AdamekMiliusEtAl12,BowlerGoncharovEtAl13}}
 This links the axiomatization and the definability parts of our
 work. For example, I/O on the signature $S$ is modelled by the free
 monad~$T_{S}$, and countable probabilistic choice by the distribution
 monad $\dist$.  There is also a monad $\vict_S$ of victorious
 strategies.  (All these are defined~\Cref{sec:prelims,sec:monads}). Via our definability result, we can formulate the 
tensor completeness problem as follows: is $\vict_S$ 
 the tensor of the monads $T_S$ and~$\dist$? \Cref{thm:comp_fin} then answers it
 positively for finitary~$S$.

\subsection{Related work}
Our approach to interpreting algebraic operations game-theoretically is closely related
to that of~\citet{Koenig21}, who considers signatures of {uninterpreted effects}
(which can be understood as I/O). Somewhat similarly to us, the author considers the question 
of definability of strategies, and establishes a characterization of non-well-founded 
strategies in deterministic setting as a final coalgebra, and as an ideal completion of a term algebra.
The work~\citet{KoenigShao20} uses this perspective to define a category of games in the 
setting of coexistent {angelic} and {demonic} non-determinism. In~\citet{BowlerLevyEtAl18}
the issue of completeness and definability in our sense is solved w.r.t.\ finitary and 
countable non-determinism. Our notion of {probabilistic strategy} agrees with
that of \citet{DanosHarmer02}, where it is used to establish full abstraction for 
probabilistic {idealized Algol}.

Tensors of theories go back to \citet{Freyd66}, who characterized tensoring with 
a large class theories, namely \emph{semiring module theories}, in the finitary case. The use of sums and tensors of monads (or theories) for semantics of computations
is elaborated by~\citet{HylandLevyEtAl07} who show that sums 
of monads need not generally exist, e.g.\ if one of the monads is the (full) powerset monad.
The existence issue is revisited by~\citet{AdamekMiliusEtAl12}, who establish 
some surprising results of existence. The general existence problem for tensors was 
solved in the negative by~\citet{BowlerGoncharovEtAl13}. Some classes of examples 
of tensor existence are given by~\citet{GoncharovSchroder11b} who also 
give some sufficient criteria for injectivity of universal maps to the tensor 
product (called ``{conservativity}'' in \emph{op.\ cit.}), what can be regarded 
as a form of our completeness problem.

Complete axiomatization of equivalence between probabilistic processes is a classical theme 
in process algebra with probabilistic choice being the only form of choice~\cite{BaetenBergstraEtAl95,GomezFrutos-EscrigEtAl97,Nunez99},
or combined with the nondeterministic choice~\cite{BandiniSegala01,TimmersGroote20,ParmaSegala04,CazorlaCuarteroEtAl03}.
The completeness issue is analyzed from the coalgebraic perspective by~\citet{SilvaSokolova11}.
The main distinction between these and our present setting is that our probabilistic choice operator 
can be infinitary, and that the signature operations can have multiple arguments.

\subsection{Structure of paper}
We start off by giving short technical preliminaries in \cref{sec:prelims}. In
\cref{sec:sys_trace}, we introduce our language of (probabilistic) traces and strategies,
instrumental for formulating and proving our results. In \cref{sec:sum-tensor},
we connect sums and tensors of theories with bisimilarity and trace equivalence.
In \cref{sec:nf} we introduce various notions of normal form for programs, as a stepping 
stone to our completeness and definability results. Those are given in \cref{sec:impers},
\cref{sec:comp-fin} and \cref{sec:victor}. We elaborate on connections to monads in
\cref{sec:monads} and draw conclusions in \cref{sec:conc}. 

\section{Preliminaries}\label{sec:prelims}
\subsection{Algebraic structures}

We begin with some operations on sets.  For a set $X$, we write $\dist X$ for 
the set of all (discrete probability) distributions on $X$, i.e.\ functions $X
\rightarrow [0,1]$ summing to~$1$, which are necessarily countably
supported.  We write $\distfin X$ for finitely
supported ones. %
Next, a \emph{convex space} (also called a convex algebra, 
barycentric space or barycentric algebra) is a set $X$ with a binary operation $+_p$ for
each $p \in (0,1)$, satisfying the equations in~\cref{fig:prob-laws}(a).
The space is \emph{cancellative} when $x \plushalf z = y
\plushalf z$ implies $x=y$.  This is equivalent to $x +_p z = y +_p z$
implying $x=y$ for all $p \in (0,1)$.

An \emph{$\omega$-convex space} (also called a superconvex space) is a
set $X$ with an operation $(x_n)_{n\in\nats} \mapsto \sum_{n\in\nats} p_n x_n$ for each
distribution $p$ on $\nats$, satisfying the equations in~\cref{fig:prob-laws}(b)
where $\delta_{n,m}$ is defined to be $1$ if $n=m$, and 0 otherwise.  
Restricting to finitely supported sequences yields an
alternative definition of convex space. 

In monad terminology, a convex space can be seen as an Eilenberg-Moore algebra for the $\distfin$
monad, and an $\omega$-convex space for the $\dist$ monad.

\subsection{Transition systems}

Let us now fix a signature $\namesig = (\arity{k})_{k
    \in\Ops}$ representing the I/O capability of our language.   For a
  set $X$, we write $\hsig X \eqdef \sum_{k\in\Ops}
  X^{\arity{k}}$---i.e.\ the set of pairs $\tuple{k,(x_i)_{i \in
      \arity{k}}}$ where $k \in K$ and $x_i \in X$ for all $i\in\arity{k}$.  We now define the notion of transition
  system, which has both nondeterministic and probabilistic versions.
\begin{definition} 
    \begin{enumerate}[wide]
    \item A \emph{nondeterministic system} consists of a set~$X$ and 
      function   $\zeta \colon X \to \pset\hsig X$.
   \item Such a system is \emph{total} when, for all $x \smin X$, the set $\zeta(x)$ is
     inhabited.
   \item Such a system is \emph{deterministic}
      (respectively, \emph{finitely
        branching}, \emph{countably
        branching}) when, for all $x \smin X$, the set  $\zeta(x)$ has at most one element 
       (respectively, is finite, is countable).
\item Given such a system: for $x \smin X$, let $x \Rrightup{k} (y_i)_{i\in\arity{k}}$ when 
      the
      set $\zeta (x)$ contains 
      $(k,(y_i)_{i\in\arity{k}})$.
    \end{enumerate}
  \end{definition}
  The transition  $x \Rrightup{k} (y_i)_{i\in\arity{k}}$ intuitively
  means that starting in
      state $x\in X$, the system may output $k$, pause
      to receive input, and then transition to state $y_i$ if the user
      inputs~$i \smin \arity{k}$.
  \begin{definition}
    \begin{enumerate}[wide]
    \item A \emph{probabilistic system} consists of a set~$X$ and
      function
      $\zeta \colon X \to \dist \hsig X$.
    \item Given such a system: for $x \smin X$, we write
      $x \Rrightup{p,k} (y_i)_{i\in\arity{k}}$ when the distribution
      $\zeta (x)$ assigns $p$ to $(k,(y_i)_{i\in\arity{k}})$.
    \item Given such a system: its \emph{support system} is the nondeterministic system
      $(X,\mathsf{supp} \circ \zeta)$, which is total and countably
      branching.  When it is finitely branching, we say that
      $(X,\zeta)$ is a 
      \emph{finitely probabilistic} system.
    \end{enumerate}
  \end{definition}
The  transition
      $x \Rrightup{p,k} (y_i)_{i\in\arity{k}}$ intuitively means that starting
      in state $x\in X$, the system, with probability $p$, will
      output~$k$, pause to receive input, and then transition to
      state~$y_i$  if the user then inputs $i \smin \arity{k}$.

In both kinds of system, the states are ``active''---they represent a
program that is running, not one that is paused and awaiting input.
  (Such systems may be called \enquote{generative},
 rather than \enquote{reactive}.)

Furthermore, the systems of interest in this paper are ``well-founded'' in the
following sense.
\begin{definition}[Well-Founded Nondeterministic Systems]\mbox{}
  \begin{enumerate}[wide]
  \item A nondeterministic system $(X,\zeta)$ is
    \emph{well-founded} when there is no infinite sequence%
    \begin{displaymath}
      x_0 \Rrightup{k_0} (y^{0}_i)_{i\in\arity{k_0}},
      \betwixt y^{0}_{i_{0}} = x_1  \Rrightup{k_1}  (y^{1}_i)_{i \in
        \arity{k_1}}, \betwixt \cdots   
    \end{displaymath}
Equivalently: $X$ is the only \emph{inductive} subset, i.e., such $U\subseteq X$ that 
any $x\smin X$ whose successors are all in
$U$ is $U$ itself.
\item A probabilistic system $(X,\zeta)$ is \emph{well-founded} when
  its support system is well-founded.
  \end{enumerate}
\end{definition}
Each of the languages mentioned in \cref{sec:axiom} gives rise to a well-founded system.  For example, take the countably
probabilistic language obtained from the signature $S$.  Writing~$\prog$ 
for the set of programs, we obtain a probabilistic system
$(\prog,\zeta)$, where the function $\zeta$ is defined by induction
over syntax: it sends
          \begin{itemize}
     \item $\inpu{k}(M_{i})_{i\in\arity{k}}$ to the
       principal distribution on $\tuple{k, (M_i)_{i\in\arity{k}}}$,
       \item $M +_p N$ to $\zeta(M) +_p \zeta(N)$,
       \item $\tesum_{n\in\nats}p_n\ml M_n$ to $\tesum_{n \in
           \nats} p_n\ml \zeta(M_n)$.
     \end{itemize}

\section{Bisimulation and traces}\label{sec:sys_trace}
\subsection{Bisimulation}

Before defining bisimulation, it is helpful to specify how~$\hsig$ and $\pset$ and
$\dist$ act on relations. For a distribution $d$ on
$A$ and any subset $U \subseteq A$, we write $d[U] \eqdef \sum_{x \in U} d(x)$.
\begin{definition}
  Let $A$ and $B$ be sets, and $\catr \subseteq A \times B$ a
  relation.
  \begin{enumerate}[wide]
  \item The relation $\hsig \catr \subseteq \hsig A \times \hsig B$
    relates $(k,(x_i)_{i\in\arity{k}})$ to  $(l,(y_i)_{i \in
      \arity{l}})$ when $k=l$ and for all $i \smin \arity{k}$ we have
    $(x_i,y_i)\in\catr$.
  \item The relation $\pset \catr \subseteq \pset A \times \pset B$
    relates $A \smin \pset A$ to $B \smin \pset B$ under the
    following condition: for every
    $x \smin A$ there is $y \smin B$ such that $(x,y)\in\catr$, and
   for every $y \smin B$ there is $x \smin A$ such that $(x,y)\in\catr$.
  \item \label{item:distrel} The relation $\dist \catr \subseteq \dist A \times \dist B$
    relates $d \smin \dist A$ to ${e \smin \dist B}$ under the following
    condition: for every 
    subset $U \subseteq A$, writing $\catr(U)$ for the set of all $\catr$-images of elements of $U$, we have $d[U] \leqslant e[\catr(U)]$.
    Equivalently: for every subset $V \subseteq B$, writing  $\catr^{\mone}(V)$ for the
    set of all $\catr$-preimages of elements of~$V$, we have $e[V]
\leqslant d[\catr^{\mone}(V)]$.
  \end{enumerate}
\end{definition}

Now we define bisimulation for nondeterministic and probabilistic
systems.\footnote{In early work~\cite{LarsenSkou91}, a
  probabilistic bisimulation was required to be an equivalence relation.}
\begin{definition}
  \begin{enumerate}[wide]
  \item Let $(X,\zeta)$ and $(Y,\xi)$ be nondeterministic transition
    systems.  A \emph{bisimulation} between them is a relation $\catr
    \subseteq X \times Y$ such that, for any $(x,y) \smin \catr$, we have
    $(\zeta(x), \zeta(y))\in\pset \hsig \catr$.
  \item Let $(X,\zeta)$ and $(Y,\xi)$ be probabilistic transition
    systems.  A \emph{bisimulation} between them is a relation $\catr
    \subseteq X \times Y$ such that, for any $(x,y) \smin \catr$, we have
    $(\zeta(x), \zeta(y))\in\dist \hsig \catr$.
  \end{enumerate}
  In each case, the greatest bisimulation is called \emph{bisimilarity}.
\end{definition}

\subsection{Plays}

As explained in the introduction, a ``play'' for our signature~$S$ 
is a alternating sequence of outputs and inputs.  To understand the
various kinds of play, imagine a program (in any of the languages in \cref{sec:axiom}) running.  When
execution has just started or an input has just 
been supplied, the program is
\emph{actively} working to produce the next output.  But after an
output has appeared, the program is paused and it \emph{passively} awaits input.  This
motivates the following terminology.
    
\begin{definition}[Plays] %
  \begin{enumerate}[wide]
  \item An \emph{active-ending play} of length ${n\in\nats}$ is a sequence of the form
    $k_0\comma i_0\comma k_1\comma i_1\comma \ldots,k_{n-1}\comma i_{n-1}$
    with $k_r\in\Ops$ and ${i_r\in\arity{k_r}}$.   An example is
    the empty play, denoted by~$\varepsilon$.
  \item A \emph{passive-ending play} $s$ of length $n\in\nats$ is a
    sequence of the form $k_0\comma i_0\comma k_1\comma i_1\comma \ldots\comma k_{n-1}\comma i_{n-1}\comma k_n$
    with $k_r\in\Ops$ and $i_r\in\arity{k_r}$.   The set of all
    possible next inputs is written  
    $\inarity{s} \eqdef \arity{k_n}$.
  \item An \emph{infinite play} is a sequence
     \begin{math}
      k_0,i_0,k_1,i_1,\ldots,
    \end{math}
    with~${k_r \in\Ops}$ and $i_r\in\arity{k_r}$.
  \end{enumerate}
\end{definition}
Individual plays can be arranged into a ``trace strategy'', which intuitively describes the trace semantics of a program.  As explained in
the introduction, this may be a set of plays (for a nondeterministic
program) or a function on plays (for a probabilistic program).

\subsection{Trace semantics for nondeterminism}

We now consider the trace semantics of a state within a
nondeterministic system.  The key notion is the following.
    \begin{definition}%
          A \emph{nondeterministic trace strategy} $\sigma$ consists
          of a set $\sigmapass$ of
          passive-ending plays and a set $\sigmaact$ of active-ending
          plays such that
          \begin{itemize}
          \item an active-ending play is in $\sigmaact$ iff it is either
            empty or of the form  $s.i$ for $s \smin \sigmapass$ and
            $i \smin \inarity{s}$
           \item for an active-ending play $s$ and output $k \smin
             \Ops$, if ${s.k\in\sigmapass}$, then $s\in\sigmaact$.
          \end{itemize}
        \end{definition}
Note that $\sigmapass$ determines $\sigmaact$.  For convenience, we
sometimes represent the strategy as the set $\sigmapass \cup \sigmaact$.

Now we can define trace semantics:
\begin{definition} \label{def:nondettrace}
  Let $(X,\zeta)$ be a nondeterministic system.  The \emph{trace
      semantics} of $x \smin X$, written $\tracesof{x}$, is the strategy
    consisting of all 
    passive-ending plays $k_0,i_0,k_1,i_1,\ldots,k_{n-1},i_{n-1}$ such
    that there is a sequence
    \begin{align*}
         x = x_0 \Rrightup{k_0} (y^{0}_i)_{i\in\arity{k_0}},
         &\betwixt y^{0}_{i_{0}} = x_1 %
      \betwixt \cdots\\ &\betwixt   y^{n-1}_{i_{n-1}} = x_{n} \Rrightup{k_n} (y^{n}_i)_{i\in\arity{k_0}}
    \end{align*}
    and likewise for active-ending plays.     Two states with the same trace semantics are \emph{trace equivalent}.
\end{definition}

    \subsection{Trace semantics for probability}\label{sec:t-sem}

 For a probabilistic system, we proceed similarly:
 \begin{definition}
A \emph{play-measure} $\sigma$ consists of a function $\sigmapass$
from passive-ending plays to $[0,1]$ and a function $\sigmaact$ from
active-ending plays to $[0,1]$ satisfying
\begin{align*}
   \sigmaact(s.i)  &\, = \sigmapass(s) &
  \sigmaact(t)     &\, = \sum_{k\in\Ops} \sigmapass(t.k) 
 \end{align*}
 The \emph{weight} of $\sigma$   %
 is $\sigmaact(\varepsilon)$.
\end{definition}
Note that $\sigmapass$ determines $\sigmaact$.  For convenience we
sometimes  represent $\sigma$ as the function $\sigmapass \cup \sigmaact$.
\begin{definition}
  A \emph{probabilistic trace strategy} is a play-measure of weight
  1.  
\end{definition}
An example of such a function is the trace semantics of the program in
\cref{sec:axiom}\plnote{More specific?} 
We proceed to define trace semantics in general:
\begin{definition} \label{def:probtrace}
  Let $(X,\zeta)$ be a probabilistic system.  The \emph{trace
       semantics} of $x \smin X$, written $\tracesof{x}$, is the probabilistic trace strategy sending  $k_0, i_0, \ldots, k_n$ to the sum over all sequences
    \begin{align*}
         x = x_0 \Rrightup{p_0,k_0} (y^{0}_i)_{i\in\arity{k_0}},
         ~~ y^{0}_{i_{0}} &\,= x_1 %
      \cdots\\ \betwixt   y^{n-1}_{i_{n-1}} &\,= x_{n} \Rrightup{p_{n},k_n} (y^{n}_i)_{i\in\arity{k_0}}
       \end{align*}
       of the product $p_0 \cdots p_{n}$.
     \end{definition}
For each of our probabilistic languages:
\begin{itemize}
\item The programs form a transition system, so each program has a
  trace semantics via \cref{def:probtrace}.
\item A strategy is
  \emph{definable} when it is the trace semantics of a program.  More
  generally, a play-measure of weight $\lambda$ is \emph{definable}
  when it is either 0 or $\lambda \tracesof{M}$ for some program
  $M$.\footnote{The first disjunct is needed only in the case of a signature
    with no constants.}
\end{itemize}
 Trace semantics can also be given compositionally, as described further below 
 in~\Cref{pro:sem-compos}. Firstly, each $k \smin \Ops$ gives an $\arity{k}$-ary operation
 $\sfinpu{k}$ on strategies:
 \begin{eqnarray*}
   \sfinpu{k}(\sigma_i)_{i \in
   \arity{k}} & : & \left\{
                    \begin{array}{llll}
                      \varepsilon & \mapsto & 1 \\
                      k & \mapsto & 1 & \\
                      k.i.s & \mapsto & \sigma_i(s) \\
                      l.s & \mapsto & 0 & (l \not= k)
                    \end{array} \right.
 \end{eqnarray*}
 Secondly, play-measures $(\sigma_i)_{i\in I}$ of total weight
 $\lambda \smin [0,1]$ can be added to give a play-measure
 \begin{spaceout}{lllll}
   \sum_{i \in I} \sigma_i & : & s & \mapsto & \sum_{i \in I} \sigma_i(s)
 \end{spaceout}%
 of weight $\lambda$.  A special case is the nullary sum, denoted 0, which 
 is the
 unique play-measure of weight 0.  Thirdly, we can multiply a scalar $b \smin [0,1]$ by a
 play-measure of weight $\lambda$ to give a play-measure
 \begin{spaceout}{lllll}
   b \ml \sigma & : & s & \mapsto & b \sigma(s)
 \end{spaceout}%
 of weight $b \lambda$.  These constructions
 yield a cancellative $\omega$-convex space of all 
 strategies.
 \begin{proposition}\label{pro:sem-compos}
 Here is a compositional description of the trace semantics:
    \begin{align*}
    \tracesof{\inpu{k}(M_{i})_{i\in\arity{k}}} =\; &
                                                       \sfinpu{k}(\tracesof{M_i})_{i \in I} \\
 \tracesof{M +_p N} 													 =\; & \tracesof{M} +_p \tracesof{N} \\
    \tracesof{\tesum_{n\in\nats}p_n\ml M_n} 			 =\; & \tesum_{n\in\nats} p_n\ml \tracesof{M_n}
  \end{align*}
\end{proposition}
For play-measures $\sigma$ and $\tau$, the notation $\sigma \leqslant
\tau$ means that $\sigma(s) \leqslant \tau(s)$ for every passive-ending
$s$.  When this is so and $w(\sigma) < \infty$, we write $\tau -
\sigma$ for the
unique play-measure $\rho$ such that $\tau = \sigma + \rho$.
Explicitly, $\rho$ sends a passive-ending play~$s$ to $\tau(s) -
\sigma(s)$.%

\subsection{Branching conditions}

\begin{definition}
  Let $\sigma$ be a nondeterministic trace strategy.
  \begin{enumerate}%
  \item For $t \smin \sigmaact$, a \emph{possible
      output} is $k \smin \Ops$ such that \mbox{$t.k\in\sigmapass$}.  
  \item We say that  $\sigma$ is
          \begin{itemize}
          \item \emph{total} when each $t \smin \sigmaact$ has at
            least one possible output
          \item \emph{finitely branching} when each  $t \smin \sigmaact$ has finitely many possible outputs
          \item \emph{countably branching} when each  $t \smin \sigmaact$-enabled has countably many possible outputs.\sgnote{What is ``enabled''?}
          \end{itemize}
  \end{enumerate}
\end{definition}
Thus, for any state of a total transition system, the trace semantics 
is total, and likewise for finite branching and countable branching.

\begin{definition}
  The \emph{support} of a play-measure $\sigma$, written
  $\supp{\sigma}$, is the nondeterministic trace strategy consisting
  of all plays $p$ such
  that $\sigma(p) > 0$.
\end{definition}
Thus $\supp{\sigma}$ is total (unless $\sigma = 0$) and countably
branching.  We say that $\sigma$ is \emph{finitely branching} if
$\supp{\sigma}$ is.

We see that, for any state of a probabilistic system,
the support of the trace semantics is the trace semantics in the
support system.  Hence, for any state of a finitely probabilistic
system, the trace semantics is finitely branching.

\section{Sum- and Tensor-Equivalence}\label{sec:sum-tensor}

Our goal is to axiomatize bisimilarity and trace equivalence. For the
latter, we note the following fact:
\begin{proposition} \label{prop:tensorholds}
The tensor laws shown in
\cref{fig:tens-laws} hold for probabilistic trace strategies. 
\end{proposition}
\begin{proof}
  Given $k \smin \Ops$, we show
  \begin{align*}
    \sum_{n\in\nats} p_n \ml \mathsf{Req}\,k?& (\sigma_{n,i})_{i \in
   \arity{k}} \\ =&\;  \mathsf{Req}\,k? \Bigl(\sum_{n\in\nats} p_n \ml  \sigma_{n,i}\Bigr)_{i \in
                     \arity{k}}
                 \end{align*}
   as follows.  Any play beginning
         with an output other than $k$ is assigned 0 by each side.
         The play $k$ itself is assigned $1$ by each side.  A play of
         the form $k.i.s$ is assigned $\sum_{n\in\nats} p_n\ml
         \sigma_{n,i}(s)$ by each side.\end{proof}

\begin{definition}
  For each probabilistic language:
  \begin{itemize}
  \item \emph{sum-equivalence} is the least congruence including the
  convex laws or $\omega$-convex laws (\cref{fig:prob-laws}).
\item \emph{tensor-equivalence} is the least congruence including the
  convex or $\omega$-convex laws, and commutativity laws between
  I/O operations and probabilistic choice (\cref{fig:tens-laws}).
  \end{itemize}
\end{definition}
The reason for the names \enquote{sum} and \enquote{tensor} is
explained in \Cref{sec:monads}.
Let us note the soundness of these theories.
\begin{proposition} \label{prop:equivsound}
  \begin{enumerate}
  \item \label{item:bissound} Sum-equivalent programs are bisimilar.
  \item \label{item:tracesound} Tensor-equivalent programs are trace equivalent.
  \end{enumerate}
\end{proposition}
\begin{proof}
  We easily see that bisimilarity validates the convex or
  $\omega$-convex laws.  We have seen that trace semantics can be
  presented compositionally, so \cref{prop:tensorholds} tells us that it validates the tensor laws.
\end{proof}

For bisimilarity, we also have completeness:
\begin{theorem} \label{pro:bisim}
  For the finitely probabilistic or countably probabilistic program on
  any I/O signature, bisimilarity coincides with sum-equivalence.
\end{theorem}
\begin{proof}
We give the case of countable probabilistic choice.  Firstly, we
  define the \emph{crude} and \emph{rooted} programs
  inductively as follows.
  \begin{itemize}[wide]
  \item The program $\inpu{k}{ (M_{i})_{i \in
          \arity{k}} }$ is rooted if the programs $(M_{i}) _{i \in
        \arity{k}} $ are all crude.
   \item The program $\sum_{n < \alpha} p_n \ml M_n$, where $\alpha$
     is either $\omega$ or a natural number, is crude if the programs
     $(M_n)_{n < \alpha}$ are all rooted.
  \end{itemize}
We prove that bisimilar crude programs are sum-equivalent, and
bisimilar rooted programs are sum-equivalent, by mutual induction.
For the rooted case, we note that
\begin{eqnarray*}
  \inpu{k}{ (M_{i})_{i \in
          \arity{k}} } & \sim & \inpu{l}{ (N_{i})_{i \in
          \arity{l}} }
\end{eqnarray*}
means that $k=l$ and, for all $i\in\arity{k}$, we have $M_i \sim
N_i$.  For the crude case, we note that
\begin{eqnarray*}
  \sum_{n < \alpha} p_n \ml M_n & \sim &  \sum_{n < \beta} q_n \ml N_n 
\end{eqnarray*}
means that, for any rooted program $R$, we have 
\begin{eqnarray*}
  \sum_{\substack{n < \alpha \\ M_n \sim R}} p_n & = & \sum_{\substack{n < \beta \\ N_n \sim R}} q_n
\end{eqnarray*}
So, up to sum-equivalence, we can gather up all the bisimilarity
classes on each side and get the same result.\end{proof}

We now consider whether the tensor theory is complete for
trace-equivalence.   In the nondeterministic setting, this question was answered
affirmatively in~\cite{BowlerLevyEtAl18}.

We shall see that tensor completeness also holds for binary
probabilistic choice, and for countable probabilistic choice if the
signature $S$ is finitary.  Henceforth, ``equivalent'' means
tensor-equivalent, and is denoted by $\equiv$.

\section{Well-founded Normal Forms and Strategies}\label{sec:nf}

To aid our analysis of trace semantics and tensor-equivalence, we shall introduce several kinds of normal form in this paper: 
\emph{finitely founded}, \emph{well-founded}, \emph{shallow}, \emph{light} and \emph{steady}.  Some of them, such as
shallow and light normal form, exist for all programs, while others exist only for
special kinds of program.

Here are the first two kinds of normal form:
\begin{definition} \label{def:wellfoundnf}
  \begin{enumerate}[wide]
  \item The \emph{finitely founded normal forms} are defined inductively
    by
    \begin{displaymath}
      M \bnfgo \tesum_{k \in A} p_k \ml \inpu{k}(M_{k,i})_{i \in
        \arity{k}}
    \end{displaymath}
    where $A$ ranges over the finite subsets of $\Ops$, and the
    coefficients~$(p_k)_{k \in A}$ are positive and sum to 1.
  \item The \emph{well-founded normal forms} are defined the same
    way, except that $A$ ranges over the countable subsets of $\Ops$.
  \end{enumerate}
\end{definition}
Strictly speaking, these definitions require a specified order on every finite subset
of $\Ops$, and an $\omega$-shaped order on every countably infinite
subset.  However, up to $\equiv$, these orders do not matter.

Now we give the corresponding conditions on strategies.
     \begin{definition} \label{prop:wnf} %
        \begin{enumerate}[wide]
     \item A nondeterministic trace strategy $\sigma$ is \emph{well-founded}
       when there is no infinite play whose prefixes 
       are all in $\sigma$.  
     \item A probabilistic trace strategy $\sigma$ is \emph{well-founded}
       when its support is well-founded.
     \item A (nondeterministic or probabilistic) strategy is \emph{finitely founded} when
       it is both well-founded and finitely branching.
        \end{enumerate}
      \end{definition}
It is evident that $M \mapsto \tracesof{M}$ is a bijective
correspondence from finitely founded normal forms to finitely founded
strategies, and from well-founded normal forms to well-founded strategies.
      \begin{theorem} \label{prop:finstuff}
        For the finitely probabilistic language on any I/O signature:
        \begin{enumerate}[wide]
        \item A play-measure is definable iff it is finitely founded.
        \item \label{item:fincomplete} Trace-equivalence coincides with tensor-equivalence.
        \end{enumerate}
     \end{theorem}
     \begin{proof}
By induction, for each program $M$, the strategy $\tracesof{M}$ is
finitely founded, and the corresponding normal form is
equivalent to $M$. Therefore, if $M$ and $M'$ have the same trace
semantics $\sigma$, then they are both tensor-equivalent to
the corresponding normal form.
     \end{proof}
This completes our study of the finitely probabilistic language.  
Henceforth we consider only the countably probabilistic one.  It is not generally the case
that every definable
      strategy is well-founded:
      \begin{example} \label{ex:notwnf}
  For the signature consisting of a unary operation~$a$ and constant $c$, the strategy 
  \begin{align*}
  \tracesof{\sum_{n\in\nats} 2^{-n-1}\ml a^n(c)}
  \end{align*}
  is not well-founded,
  as this strategy assigns a positive probability to each finite prefix of the infinite play $(a.0)^\omega$, i.e.\ to $\epsilon$, $a$, $a.0$, $a.0.a$, $a.0.a.0$, etc.
\end{example}
Nonetheless, we have the following result.
\begin{proposition}  \label{prop:wfstratdef} 
  \begin{enumerate}[wide]
\item \label{item:wfstratdef} Any well-founded play-measure is definable.
  \item \label{item:wfcomplete} For trace-equivalent programs $M$ and $N$ such that $\tracesof{M}$
    is well-founded, we have $M \equiv N$.
  \end{enumerate}
      \end{proposition}
      \begin{proof}
        By induction, for each well-founded normal form $M$, we see
        that $M$ is equivalent to any program that is trace-equivalent
        to it.  So programs $N$ and $N'$ that denote the same well-founded
        strategy are both equivalent to the corresponding
        normal form.
      \end{proof}
      Next we consider a weaker notion than well-founded normal form, in which the components $(M_{k,i})_{k \in A, i \in
  \arity{k}}$ are allowed to be arbitrary programs.
      \begin{definition} \label{def:shallow}
        A \emph{shallow normal form} is a program of the form
        $\tesum_{k \in A} p_k \ml
        \inpu{k}(M_{k,i})_{i\in\arity{k}}$ with countable $A\subseteq\Ops$, and the
    coefficients $(p_k)_{k \in A}$ are positive and sum to 1.
  \end{definition}
\begin{proposition}
  Every program is equivalent to a shallow normal form.
\end{proposition}

Now we can give the following result,
which will be used in \cref{sec:impers}.
      \begin{lemma} \label{prop:subsplit}
        For programs $M$ and $L$ with $\tracesof{L}$ well-founded, and
        $p < 1$ such that $p\tracesof{L}
\leqslant \tracesof{M}$, there
is a program $N$ such that $M \equiv L +_p N$.  
      \end{lemma}
\begin{proof}
 We assume that $M = \tesum_{k \in A} a_k \ml
        \inpu{k}(M_{k,i})_{i\in\arity{k}}$ is in shallow normal-form and that $L = \tesum_{k \in B} b_k \ml
        \inpu{k}(L_{k,i})_{i\in\arity{k}}$ is in 
        well-founded normal form.  We proceed by induction on $L$.

For $k \smin B$, the inequality $p \tracesof{L} \leqslant
\tracesof{M}$ tells us that
\begin{align*}
  0 < p b_k  
    = p \tracesof{L}(k) 
    \leqslant \tracesof{M}(k) 
\end{align*}
giving $k \in A$ and $pb_k \leqslant a_k$.  Furthermore, for all $i \smin
\arity{k}$ we have $\frac{p b_k}{a_k} \tracesof{L_{k,i}} \leqslant
\tracesof{M_{k,i}}$, since $a_k > 0$ and for any passive-ending $s$ we have
\begin{flalign*}
&& p \tracesof{L} (k.i.s) \leqslant &\; \tracesof{M}(k.i.s) && \\
  \text{i.e.} && p b_k \tracesof{L_{k,i}}(s) \leqslant &\; a_k
                                                              \tracesof{M_{k,i}}(s)
                                                            \end{flalign*}%
The inductive hypothesis gives $N_{k,i}$ such that 
  \begin{align*}
    M_{k,i} =\, & {L_{k,i}} +_{\frac{pb_k}{a_k}} N_{k,i}
  \end{align*}                                      
  Put 
  \begin{align*}
    N \eqdef\, & \sum_{k \in A} \begin{cases}
                 \frac{a_k - pb_k}{1-p}\ml  \inpu{k}(N_{k,i})_{i\in\arity{k}} & (k \in B) \\
                  \frac{a_k}{1-p}\ml \inpu{k}(M_{k,i})_{i\in\arity{k}}        & (k \not\in B)
                \end{cases}
  \end{align*}
 Clearly the coefficients are nonnegative and sum to 1.   Now we have
  \begin{align*}
    M \equiv\, & \sum_{k \in A} a_k\ml\inpu{k}(M_{k,i})_{i \in
                 \arity{k}} \\
      \equiv\, & \sum_{k \in A} 
                 \begin{cases}
                  a_k\ml\inpu{k} (L_{k,i}
                   +_{\frac{pb_k}{a_k}} N_{k,i})_{i\in\arity{k}} &
                                                                     (k
                                                                     \in
                   B) \\
                  a_k\ml\inpu{k}(M_{k,i})_{i\in\arity{k}} & (k \not\in B)
                 \end{cases} \\
      \equiv\, & \sum_{k \in A} \begin{cases}
                  a_k\ml \bigl(\inpu{k} (L_{k,i})_{i\in\arity{k}}\\
                   \qquad +_{\frac{pb_k}{a_k}} \inpu{k}(N_{k,i})_{i \in
                   \arity{k}}\bigr) &~\quad
                                                                     (k
                                                                     \in
                   B) \\
                  a_k\ml\inpu{k}(M_{k,i})_{i \in
                 \arity{k}} &\quad~ (k \not\in B)
                 \end{cases} \\
     \equiv\, & \sum_{k \in B} b_k\ml \inpu{k}  (L_{k,i})_{i
               \in\arity{k}}\\&\quad +_p \sum_{k \in A} 
                \begin{cases}
                 \frac{a_k - pb_k}{1-p}\ml  \inpu{k}(N_{k,i})_{i \in
                  \arity{k}} & (k \in B) \\
                  \frac{a_k}{1-p}\ml \inpu{k}(M_{k,i})_{i\in\arity{k}}
                             & (k \not\in B)
                \end{cases} \\
    =\, & L +_p N
  \end{align*}
  as required.\end{proof}            

\section{Completeness up to Impersonation}\label{sec:impers}

Although it is an open question whether trace
equivalent programs are always 
tensor-equivalent, we shall show that tensor-equivalence is complete
in a weaker sense.  This uses the
following notion.  
\begin{definition}
  Let $A$ be a convex space, and $a= (a_n)_{n\in \nats}$ a sequence in $A$.
  \begin{enumerate}[wide]
  \item An \emph{$a$-solution} is a sequence  $(r_n)_{n\in\nats}$ in
    $A$ satisfying $r_n = a_{n} +_{\frac{1}{2}} r_{n+1}$ for all $n
    \smin \nats$.
  \item If $A$ is an $\omega$-convex space, then the \emph{canonical
       $a$-solution} is the sequence $\bigl(\sum_{i\in\nats}
     2^{-n-i-1} \ml  a_{n+i}\bigr)_{n\in\nats}$.
  \end{enumerate}
\end{definition}
The \emph{head} of a sequence  $(r_n)_{n\in\nats}$ is the element
$r_0$.  
\begin{definition}
  Let $A$ be an $\omega$-convex space with elements~$x$ and $y$.  We
  say that $x$ is \emph{impersonated} by $y$ when there is a sequence
  $a$ such that
  \begin{itemize}
  \item $x$ is the head of the canonical $a$-solution, i.e.\ $x =
    \sum_{i\in\nats} 2^{-i-1} a_i$.
  \item $y$ is the head of an $a$-solution.
  \end{itemize}
\end{definition}

We see that impersonation implies equivalence up to cancellation:
\begin{lemma}\label{lem:impers-cancel}
  Let $A$ be an $\omega$-convex space with elements~$x$ and~$y$.  If
  $x$ is impersonated by $y$, then there is $u \smin A$ such that $x
  +_{\frac{1}{2}} u = y
  +_{\frac{1}{2}} u$.
\end{lemma}
\begin{proof}
  There is a sequence $a= (a_n)_{n\in \nats}$ and an $a$-solution $r$ such that $x =
  \sum_{i\in\nats} 2^{-i-1} \ml a_i$ and $y = r_0$.  Then
  \begin{align*}
    y  \plushalf& \sum_{i\in\nats} 2^{-i-1} \ml r_{i+1}
    \\
     =&\, r_0 \plushalf
                                                      \sum_{i \in
                                                      \nats}
                                                              2^{-i-1} \ml r_{i+1}
    \\
                                                =&\, \sum_{i \in
                                                   \nats} 2^{-i-1}
                                                   \ml  r_i \\
                                                =&\, \sum_{i \in
                                                   \nats} 2^{-i-1}
                                                   \ml (a_i \plushalf r_{i+1}) \\
    =&\, \Bigl(\sum_{i\in\nats} 2^{-i-1} \ml a_i\Bigr) \plushalf \Bigl(\sum_{i \in
       \nats} 2^{-i-1} \ml 
          r_{i+1}\Bigr) \\
    =&\, x \plushalf \sum_{j\in\nats} 2^{-j-1} \ml 
          r_{j+1} \tag*{$\qed$}
   \end{align*}
\noqed\end{proof}
We shall use yet another kind of normal form.
\begin{definition}
  A \emph{light pre-normal form} is a program of the form  $\sum_{n
   \in\nats} p_n \ml M_n$, where the components $(M_n)_{n\in\nats}$
use only binary probabilistic choice.  It is a \emph{light normal
  form} when the sequence $p$ is $(2^{-n-1})_{n\in\nats}$.
\end{definition}

\begin{lemma} \label{prop:light}
  Every program is light normalizable, i.e.\ equivalent to a light
  normal form.
\end{lemma}
\begin{proof}
  Any light pre-normal  form is light normalizable, by
  finitary splitting and regathering. Light normalizability is
  preserved by $\mathsf{Req}\,k?$ via the tensor law, obviously
  preserved by $+_p$, and preserved by countable
  probabilistic choice via the following argument.   Given $M_i \equiv 2^{-n-1} \ml M_{i,n}$
  for all $i \smin \nats$, and a sequence $(p_i)_{i\in\nats}$ of
  positive reals summing to~1, the program $\sum_{n\in\nats} p_n \ml
  M_n$  has light pre-normal form $\sum_{r\in\nats} \beta_r \ml N_r$,
  writing
  \begin{align*}
    \beta_r  \eqdef & \;\sum_{i \leqslant r} p_i^{i-r-1} \\
    N_r \eqdef & \; \sum_{j \leqslant r} \frac{p_i
                 2^{i-r-1}}{\beta_r} \ml M_{i,r-i}
  \end{align*}
  The result follows by induction over programs.
  \end{proof}

\begin{theorem}[Completeness up to impersonation]\label{thm:upto-imp}
Let~$M$ and $N$ be trace-equivalent programs.  Then $M$ is
impersonated by $N$, up to tensor-equivalence.
\end{theorem}
\begin{proof}
 \cref{prop:light} tells us that $M$ is equivalent to a light normal form
  $\sum_{n\in\nats} 2^{-n-1}\ml P_n$, and therefore (up to $\equiv$) it is the head
  of the canonical $P$-solution.  For each $n \smin \nats$, we define a program~$Q_n$ trace-equivalent
to $\sum_{i\in\nats} 2^{-i-1}\ml P_{n+i}$ as follows.
  \begin{align*}
    Q_0   \eqdef &\; N \\
    Q_{n} \equiv &\; P_{n} +_{\frac{1}{2}} Q_{n+1} \betwixt \text{
                     by \cref{prop:subsplit}.}
  \end{align*}
  So (up to $\equiv$) we see that $N$ is the head of the $P$-solution~$Q$.
\end{proof}
Using \cref{lem:impers-cancel}, we obtain
\begin{corollary}[Completeness up to cancellation]\label{cor:up-to-cancel}
  Let $M$ and~$N$ be trace-equivalent programs.  Then there is a
  program~$U$ such that $M +_{\frac{1}{2}} U \equiv N +_{\frac{1}{2}} U$.
\end{corollary}
\section{Completeness for Finitary Signatures}\label{sec:comp-fin}

This section is devoted to proving that, for a \emph{finitary} signature $S = (\arity{k})_{k
  \in K}$,  any two trace-equivalent programs are tensor-equivalent.
We begin with the
following notion.
\begin{definition}
  Let $\sigma$ and $\tau$ be play-measures with \mbox{$w(\tau) < \infty$}.  We say that $\sigma$
  is \emph{uniformly below} $\tau$, written $\sigma \prec \tau$, when there is $d > 0$ such
  that
\begin{align*}
    \sigma(s) + d \leqslant\; & \tau(s) \quad\quad \text{ for all $s
                                \smin \supp{\sigma}$.}
\end{align*}
\end{definition}
Thus $\sigma \prec \tau$ implies $\sigma \leqslant \tau$ and
$w(\sigma) < w(\tau)$.  The key property of $\prec$ is the following.
\begin{lemma} \label{prop:prec}
  Let $(\tau_i)_{i\in\nats}$ be a sequence of play-measures with
  total weight $\lambda \in [0,1]$.  For any play-measure $\sigma$ such that $\sigma \prec \sum_{i\in\nats} \tau_i$, there is $n \smin
  \nats$ such that $\sigma \prec \sum_{i < n} \tau_i$.  
\end{lemma}
\begin{proof}
Take a witness $d \smin (0,1]$ for $\sigma \prec \sum_{i\in\nats}
\tau_i$.  Take $n \smin \nats$ such that $\sum_{i \geqslant n} w(\tau_i)
  \leqslant \frac{1}{2}d$.  We show that $\frac{1}{2}d$ witnesses $\sigma \prec \sum_{i
    < n} \tau_i$.   For any $s\in\supp{\sigma}$, we have
   \begin{align*}
     \sigma(s) + d \leqslant\,& \sum_{i \in\nats}\tau_i(s) \\
     =\,& \sum_{i  < n}\tau_i(s)  + \sum_{i   \geqslant n}\tau_i(s)
     \\
     \leqslant\,& \sum_{i < n}\tau_i(s)  + w\Bigl(\sum_{i   \geqslant
                   n}\tau_i(s)\Bigr) \\*
      \leqslant\,& \sum_{i   < n}\tau_i(s)  + \frac{1}{2}d,
      \intertext{therefore}
      \sigma(s) + \frac{1}{2}d \leqslant\,& \sum_{i   < n}\tau_i(s)\tag*{$\qed$}
   \end{align*}
\noqed \end{proof}
Now we give the central notion of the section:
\begin{definition}
  A \emph{steady pre-normal form} is a program of the form $\sum_{n \in
    \nats} p_n \ml M_n$,  where
  \begin{itemize}
  \item for all $n \smin \nats$, the strategy $\tracesof{M_n}$ is
    well-founded
  \item for all $m \smin \nats$, we have $\sum_{n < m} p_n \ml \tracesof{M_n}
    \prec \sum_{n \in\nats} p_n \ml \tracesof{M_i}$. 
  \end{itemize}
  Such steady pre-normal form is a \emph{steady normal form} when the sequence $p$ is
  $(2^{-n-1})_{n\in\nats}$.
\end{definition}

\begin{lemma}\label{lem:steady-pre}
  Any steady pre-normal form is equivalent to a steady normal form.
\end{lemma}
\begin{proof}
  By splitting and regathering.
\end{proof}

\begin{lemma}\label{lem:main-lem}
  Let $M$ and $N$ be steady-normalizable programs.  If they are
  trace-equivalent, they are
  logically equivalent.
\end{lemma}
\begin{proof}
We have steady normal forms
  \begin{align*}
    M \equiv\, & \sum_{i\in\nats} 2^{-i-1} \ml M_i \\
    N \equiv\, & \sum_{i\in\nats} 2^{-i-1} \ml N_i
  \end{align*}
  and $\tracesof{M} = \tracesof{N}$.  \cref{prop:prec} tells us that
  \begin{itemize}[wide]
  \item for all $m \smin \nats$, there is $n \smin \nats$ such that  
  \begin{align*}
  \sum_{i < m}
   2^{-i-1} \tracesof{M_i} \prec \sum_{i < n}  2^{-i-1} \tracesof{N_i}
  \end{align*}
  \item for
  all $n \smin \nats$, there is $m \smin \nats$ such that 
  \begin{align*}
 \sum_{i < n}  2^{-i-1} \tracesof{N_i} \prec \sum_{i < m}  2^{-i-1} \tracesof{M_i} 
  \end{align*}
\end{itemize}
Thus we obtain strictly increasing sequences $(m_r)_{r\in\nats}$ and~$(n_{r})_{r
 \in\nats}$ of natural numbers satisfying $m_0 = 0 $ and
\begin{align*}
  \sum_{i < m_r}2^{-i-1}& \tracesof{M_i}\\ \prec &\; \sum_{i < n_r} 2^{-i-1}\tracesof{N_i}
  \\
  \prec &\; \sum_{i < m_{r+1}} 2^{-i-1}\tracesof{M_i}
  \end{align*}
For each $r \smin \nats$, we obtain a play-measure 
\begin{align*}
\sum_{i<n_r}2^{-i-1} \tracesof{N_i}  - \sum_{i <
    m_r} 2^{-i-1} \tracesof{M_i},
\end{align*}
which is well-founded since
  $\sum_{i<n_r} 2^{-i-1}\tracesof{N_i}$ is (being a finite sum of well-founded
  play-measures).  So, by \cref{prop:wfstratdef}(\ref{item:wfstratdef}), it is expressible as $a_r\tracesof{U_r}$
  for some $a_r >0$ and 
  program $U_r$.  Likewise the play-measure 
  \begin{align*}
  \sum_{i <
      m_{r+1}}2^{-i-1}\tracesof{M_i} - \sum_{i<n_r}2^{-i-1}\tracesof{N_i} 
  \end{align*}
  is expressible as
    $b_r \tracesof{V_r}$  for some $b_r > 0$ and
    program~$V_r$.   Now
  \cref{prop:wfstratdef}(\ref{item:wfcomplete}) gives the following.

\begin{flalign*}
\sum_{i \in [m_r \twodots m_{r+1})} \ofrac{2^{-i-1}}{a_r +
         b_r} \ml M_i
                                                     \equiv&\;
     \ofrac{a_r}{a_r + b_r} \ml U_r + \ofrac{b_r}{a_r + b_r} \ml  V_r &&(r \smin \nats)\\
\sum_{i \in [0 \twodots n_0)} \ofrac{2^{-i-1}}{a_0}\ml N_i \equiv&\;  U_0\\
\sum_{i \in [n_r \twodots
     n_{r+1})}\ofrac{2^{-i-1}}{a_r + b_{r+1}} \ml N_i \equiv&\;
   \ofrac{a_r}{a_r + b_{r+1}} \ml V_r + \ofrac{b_{r+1}}{a_r + b_{r+1}}
   \ml U_{r+1}&&(r \smin \nats)
\end{flalign*}
 So we have
  \begin{align*}
    M \equiv&\, \sum_{i\in\nats}2^{-i-1} M_i \\
      \equiv&\,   \sum_{r\in\nats} \sum_{i \in [m_r \twodots m_{r+1})}
                 2^{-i-1} M_i \\
      \equiv&\, \sum_{r\in\nats} (a_r U_r+ b_r V_r) \\
      \equiv&\, a_0 U_0 + \sum_{r\in\nats} (b_rV_r + a_{r+1}U_{r+1}) \\
      \equiv&\,  \sum_{i \in [0 \twodots n_0)} 2^{-i-1} N_i \\&\qquad + \sum_{r \in
               \nats} \sum_{i \in [n_r \twodots n_{r+1})} 2^{-i-1} N_i \\
      \equiv&\, \sum_{i\in\nats}2^{-i-1} N_i \\*
      \equiv&\, N
  \end{align*}
as required.
\end{proof}

Let us see how to obtain steady-normalizable programs.
\begin{lemma} %
  \begin{enumerate}[wide]
  \item If $M$ and $N$ are steady-normalizable, then so is $M +_p N$.
  \item If $M_n$ is steady-normalizable for all $n \smin \nats$, then
    so is $\sum_{n\in\nats}p_n \ml M_n$.
     \item Let $k \smin K$ have arity $n \smin \nats$.  If $M_j$ is steady-normalizable
    for all $j < n$, then so is $\inpu{k}(M_j)_{j < n}$.
  \end{enumerate}
\end{lemma}
\begin{proof}
  \begin{enumerate}[wide]
  \item We have steady normal forms
    \begin{align*}
      M \equiv&\; \sum_{i\in\nats}2^{-i-1} \ml M_i \\
      N \equiv&\; \sum_{i \in
                   \nats} 2^{-i-1} \ml N_i 
   \intertext{This gives } 
   M +_p N \equiv&\; \sum_{i\in\nats}2^{-i-1} \ml (M_i +_p N_i) 
    \end{align*}
    which is a steady normal form.
  \item We may assume all $p_n$ are positive.  We have steady normal forms
   \begin{align*}
       M_n \equiv&\;
   \sum_{i\in\nats} 2^{-i-1} \ml M_{n,i} \quad\quad \text{ for all $n \smin
                     \nats$.} 
     \intertext{This gives } 
   \sum_{n\in\nats}p_n\ml M_n \equiv&\;
                                        \sum_{r\in\nats} \beta_r \ml
                                        N_r \\
     \text{writing } \beta_r  \eqdef & \sum_{j \leqslant r}
                                        {p_{j}2^{-(r-j+1)}} \\
 N_r  \eqdef & \sum_{j \leqslant r} \frac{p_{j}2^{-(r-j+1)}}{\beta_r}\ml
                                        M_{j,r-j}
   \end{align*}
We prove that this is a steady pre-normal form, which is sufficient by \cref{lem:steady-pre}. 
Since $+_p$ preserves well-foundendess, we see that $\tracesof{N_r}$ is well-founded.  %
 Given $n \smin
\nats$, we want a
witness $d$ for
\begin{align*}
  \sum_{r < n} \beta_r \ml
                                        \tracesof{N_r}  \prec &\;
                                                                  \sum_{j
                                                                  \in
                                                                  \nats}
                                                                  \beta_r
                                                                  \ml \tracesof{N_r}                                          
\end{align*}
The left-hand side is equal to $\sum_{j < n} p_j \tau_j$ with $\tau_j
\eqdef \sum_{i < n-j} 2^{i} \tracesof{M_{j,i}}$, and the right-hand side to 
$\sum_{j\in\nats} p_j \sigma_j$ with $\sigma_j \eqdef \tracesof{M_j}$.  For each $j < n$,
let $d_j$ be a witness for $\tau_j \prec \sigma_j$.  We show that $d
\eqdef \min(\setbr{1}\cup
                                          \setbr{p_j d_j \mid j
                                            < n})$ has the desired
                                          property.  Any~$s$ in the
                                          support of~$\sum_{j < n} p_j
                                          \tau_j$ must be in the
                                          support of~$\tau_m$ for some
                                          $m < n$.  So we have
                                          \begin{align*}
                                            \Bigl(\sum_{j < n} p_j \tau_j\Bigr)(s)
                                            + d  \leqslant &\;
                                            \sum_{j < n} p_j
                                                        \begin{cases}
                                                          \tau_j(s) +
                                                          d_j &
                                                              \text{($j
                                                              = m$)}
                                                          \\
                                                          \tau_j(s) & \text{otherwise}
                                                        \end{cases} \\
                                             \leqslant &\; \Bigl(\sum_{j\in\nats}
                                                         p_j \sigma_j\Bigr)(s)
                                          \end{align*}
as required.
  \item %
    For all $j<n$, we have a steady
    normal form
    \begin{flalign*}
    &&  M_j \equiv&\;  \sum_{i\in\nats} 2^{-i-1} \ml M_{j,i} & \text{ for
                     all $j < n$.} 
    \intertext{This gives }
    && \inpu{k}(M_j)_{j < n)} \equiv&\;\sum_{i\in\nats}
                                                 2^{-i-1} \ml N_i \\
    &&  \text{writing}\quad N_i \eqdef &\;  \inpu{k}(M_{j,i})_{j < n)} 
    \end{flalign*}
We claim that this is  steady normal form.  Since $+_p$ preserves
well-foundendess, we see that $\tracesof{N_i}$ is well-founded.  Given $m \smin \nats$, we want a witness $d$ for
\begin{align*}
\sum_{i < m} 2^{-i-1} \tracesof{N_i} \prec&\; \sum_{i\in\nats} 2^{-i-1}
\tracesof{N_i} 
\intertext{For $j < n$, obtain a witness $d_j$ for}
\sum_{i < m} 2^{-i-1} \tracesof{M_{j,i}} \prec&\; \sum_{i\in\nats} 2^{-i-1}
\tracesof{M_{j,i}} 
\end{align*}
We show that $d \eqdef \min (\setbr{2^{-m}} \cup \setbr{d_j \mid j <
  n})$ has the desired property.  Any play in the support of $\sum_{i <
  m} 2^{-i-1} \tracesof{N_i}$ is either empty or $k$ or of the form
$k.j.s$.  For the case of $k$, we reason
\begin{align*}
  \Bigl(\sum_{i < m}& 2^{-i-1} \tracesof{N_i}\Bigr) (k) + d \\
  \leqslant &\; \Bigl(\sum_{i < m} 2^{-i-1}\Bigr) + 2^{-m} \\
  = &\; \sum_{i\in\nats} 2^{-i-1} \\
  = &\; \Bigl(\sum_{i\in\nats} 2^{-i-1} \tracesof{N_i}\Bigr)(k)
\end{align*}
and likewise for the empty play.  For the case of $k.j.s$ with \mbox{$j <n$}, 
the play $s$ is in the support of $\sum_{i < m}2^{-i-1}
\tracesof{M_{j,i}}$, so we have
\begin{align*}
  \Bigl(\sum_{i < m} 2^{-i-1}& \tracesof{N_i}\Bigr) (k.j.s) + d \\
  \leqslant &\; \Bigl(\sum_{i < m} 2^{-i-1} \tracesof{N_i}\Bigr) (s) + d_j \\ 
  \leqslant &\; \Bigl(\sum_{i\in\nats} 2^{-i-1} \tracesof{N_i}\Bigr) (s) \\
  = &\;  \Bigl(\sum_{i\in\nats} 2^{-i-1} \tracesof{N_i}\Bigr) (k.j.s)
\end{align*}
as required.\qed
  \end{enumerate}\noqed
\end{proof}

\begin{corollary}\label{cor:fin-stead-norm}
  For a finitary signature, every program is steady-normalizable.
  \end{corollary}

The finitarity hypothesis in \cref{cor:fin-stead-norm} cannot be removed:
\begin{proposition}\label{pro:inf-example}
  For the signature consisting of an $\omega$-ary operation $b$ and constants
  $(c_n)_{n\in\nats}$, let $M$ be the program $\inpu{b}\bigl(\sum_{i \leqslant n}
  \frac{1}{n+1} c_i\bigr)_{n\in\nats}$. 
  \begin{enumerate}[wide]
  \item No play-measure $\sigma$ other than 0 satisfies $\sigma \prec
    \tracesof{M}$.
  \item $M$ is not steady-normalizable.
  \end{enumerate}
\end{proposition}
\begin{proof}
 \begin{enumerate}[wide]
  \item Take a witness $d \smin (0,1]$ for $\sigma
    \prec \tracesof{M}$, and $n \smin \nats$ large enough that $n
  \geqslant \frac{1}{d}$, or equivalently $\frac{1}{n} \leqslant d$.
  We have
  \begin{align*}
    0 <&\, w(\sigma) 
     = \sigma(\varepsilon) 
     = \sigma(b) 
     = \sigma(b.n) 
     = \sum_{i \leqslant n} \sigma(b.n.c_i)
  \intertext{
  So there is $i \leqslant n$ such that $b.n.c_i\in\supp{\sigma}$, and therefore}
    d \leqslant&\,\sigma(b.n.c_i) + d 
      \leqslant \tracesof{M}(b.n.c_i) 
      = \frac{1}{n+1} 
      < \frac{1}{n}
  \end{align*}
  contradicting  $\frac{1}{n} \leqslant d$.
  \item Follows.\qed
  \end{enumerate}
  \noqed\end{proof}
\cref{lem:main-lem} and \cref{cor:fin-stead-norm} yield our main result:
\begin{theorem}\label{thm:comp_fin}
  For countably probabilistic programs on a finitary I/O signature, tensor-equivalence coincides with  trace-equivalence. 
\end{theorem}

 \section{Definable $=$ Victorious}\label{sec:victor}     

To characterize the definable strategies, the following is helpful.
\begin{definition}
 A \emph{partial counterstrategy} $\rho$ consists of a set $\rhoact$ of
 active-ending plays and a set $\rhopass$ of passive-ending plays
 satisfying the following:
 \begin{itemize}
   \item We have $s\in\rhopass$ if $s$ is of the form $t.k$ for $t \smin
     \rhoact$ and~$k \smin \Ops$.
   \item We have $\varepsilon\in\rhoact$.
   \item If $s.i\in\rhoact$, then $s\in\rhopass$.
    \item Any $s \smin \rhopass$ has at most one \enquote{prescribed input}, i.e., $i \smin \inarity{s}$
      such that $s.i\in\rhoact$.
    \end{itemize}
    A play $s \smin \rhopass$ with no prescribed inputs is called 
    a \emph{$\rho$-failure}, and the set of all such is written
    $\failof{\rho}$.  For any $m \smin \nats$, we write~$\rhoactof{m}$
    for the set of plays in $\rhoact$ of length $m$, and likewise~$\rhopassof{m}$ and $\rhofailof{m}$.  
  \end{definition}
Intuitively, $\rho$ deterministically tells the user what input to
provide, if any.  Although  $\rhopass$ is determined by $\rhoact$ and therefore redundant, including it is convenient.   %

\begin{definition}
   Let $\sigma$ be a play-measure and $\rho$ a partial counterstrategy.  For any $m
  \smin \nats$, put
  \begin{eqnarray*}
    \contprob{m}{\rho}{\sigma} &\eqdef  & \sum_{s \in
                                          \rhoactof{m}} \sigmaact(s) %
  \end{eqnarray*}
\end{definition}
In the case that $\sigma$ is a strategy, we think of $\contprob{m}{\rho}{\sigma}$ as the probability that play between
$\sigma$ and $\rho$ lasts for at least $m$ output-input cycles.
\begin{lemma} \label{prop:decsurvive}
  For any play-measure $\sigma$ and  partial counterstrategy $\rho$,
  we have
  \begin{align}
    \label{eq:zerocstrat}
    \contprob{0}{\rho}{\sigma} =&\; w(\sigma) \\
      \label{eq:succcstrat}
      \contprob{m}{\rho}{\sigma} =&\; \sum_{s\in\rhofailof{m}} \sigmapass(s)
                               + \contprob{m+1}{\rho}{\sigma} \quad \text{
                                       for $m \smin \nats$}
   \end{align}
 \end{lemma}
 \begin{proof}
   Equation \eqref{eq:zerocstrat} is obvious, and for \eqref{eq:succcstrat} we have
   \begin{align*}
       \contprob{m}{\rho}{\sigma}
    =&\;   \sum_{s\in\rhoactof{m}} \sum_{k
                     \in K} \sigmapass(s.k) \\
    =&\;   \sum_{s\in\rhoactof{m}} \sum_{k
          \in K}  
          \begin{cases}
            \sigmapass(s.k) 	& \text{ if $s.k\in\failof{\rho}$} \\
            \sigmaact(s.k.i) 	& \text{ if $s.k.i\in\rhoact$}
          \end{cases} \\
    =&\; \sum_{s\in\rhofailof{m}} \sigmapass(s)
                               + \contprob{m+1}{\rho}{\sigma} \tag*{$\qed$}
  \end{align*}
\noqed\end{proof}
For a given play-measure $\sigma$, we thus obtain a decreasing sequence
\begin{displaymath}
  w(\sigma)  =  \contprob{0}{\rho}{\sigma} \geqslant \contprob{1}{\rho}{\sigma} \geqslant \cdots
\end{displaymath}
whose infimum we call  $\continf{\rho}{\sigma}$.  In the case that
$\sigma$ is a strategy, we think of this as the
probability that play between $\rho$ and $\sigma$ lasts forever.

\begin{lemma}
  For any play-measure $\sigma$ and partial counterstrategy $\rho$, we have
  \begin{align*}
    \continf{\rho}{\sigma} =&\; w(\sigma) -  \sum_{s
                               \in
                               \failof{\rho}} \sigmapass(s)
  \end{align*}
\end{lemma}
\begin{proof}
  Follows from the equation
  \begin{align*}
    \continf{m}{\rho}{\sigma} =&\;  w(\sigma) -  \sum_{\substack{s
                               \in
                               \failof{\rho} \\ |s|
                               <
                               m}} \sigma(s)
  \end{align*}
 which we prove by induction on $m$ using \cref{prop:decsurvive}.
\end{proof}

\begin{definition}
  A play-measure $\sigma$ is \emph{victorious} when, for every
  partial counterstrategy $\rho$, we have
  \begin{spaceout}{lrcl}
  &  \continf{\rho}{\sigma} & = &  0 \\
    \text{or equivalently} & \sum_{s\in\failof{\rho}} \sigmapass(s) & =
    & w(\sigma) 
  \end{spaceout}%
\end{definition}

\begin{lemma} \label{prop:presvictor} %
  \begin{enumerate}[wide]
  \item For victorious strategies $\sigma$ and $\sigma'$, the strategy~$\sigma +_p \sigma'$ is victorious.
  \item For victorious strategies $(\sigma_n)_{n\in\nats}$, the
    strategy~$\sum_{n\in\nats} p_n \sigma_n$ is victorious.
  \item For $k \smin K$ and victorious strategies $(\sigma_i)_{i \in
      \arity{k}}$, the strategy $\sfinpu{k}(\sigma_i)_{i \in I}$ is victorious.
  \end{enumerate}
    \end{lemma}
\begin{proof}
      Let $\rho$ be a partial counterstrategy.
      \begin{enumerate}[wide]
      \item  For all $m \smin \nats$ we have
        \begin{math}
\contprob{m}{\rho}{\sigma +_p \sigma'}  =  p \contprob{m}{\rho}{\sigma} + {(1-p)}
                                      \contprob{m}{\rho}{\sigma}
        \end{math}.  
        Therefore
        \begin{align*}
          \continf{\rho}{\sigma +_p \sigma'} =&\; p \continf{\rho}{\sigma} + (1-p)
                                             \continf{\rho}{\sigma} 
          = 0
        \end{align*}
      \item Similar.
      \item If $k\in\failof{\rho}$, then for all $m \smin \nats$ we
        have $P^{m+1}_{\rho}(\sfinpu{k}(\sigma_i)_{i \in I}) = 0$.
      Therefore $P_{\rho}(\sfinpu{k}(\sigma_i)_{i \in I}) = 0$.  On
    the other hand, if  $k.i \smin \rho$, then let $\rho/k$ be the
    partial counterstrategy consisting of all plays
    $s$ such that $k.i.s\in\rho$.  For all $m \smin \nats$ we have
    $\contprob{m+1}{\rho}{\sfinpu{k}(\sigma_i)_{i \in I}} =
  \contprob{m}{\rho/k}{\sigma_i}$.  Therefore $\continf{\rho}{\sfinpu{k}(\sigma_i)_{i \in I}} = 
                                                    \continf{\rho/k}{\sigma_i} 
    = 0$.
\qed
     \end{enumerate}\noqed\end{proof}

    \begin{lemma} \label{prop:victoreq} For a play-measure $\sigma$ of weight $w$, the following
      are equivalent.
      \begin{enumerate}[wide]
      \item \label{item:vict} $\sigma$ is victorious.
      \item \label{item:supff}  For all $x \smin [0,w)$, there is a finitely founded
        play-measure $\tau \leqslant \sigma$ of weight greater or equal $x$.
      \item \label{item:supvict} For all $x \smin [0,w)$, there is a victorious 
        play-measure $\tau \leqslant \sigma$ of weight greater or equal $x$.
      \end{enumerate}
   \end{lemma}
   \begin{proof}
For (\ref{item:supff}) $\Rightarrow$ (\ref{item:supvict}), we note
that any well-founded play-measure is victorious since it is definable.

For (\ref{item:supvict}) $\Rightarrow$ (\ref{item:vict}), let $\rho$
be a counterstrategy.  For any $\varepsilon > 0$, we show that
$\continf{\rho}{\sigma} \leqslant 2\varepsilon$.  Take a victorious play-measure
$\tau \leqslant \sigma$ of weight greater or equal $w(\sigma) - \varepsilon$.  For all $n \smin \nats$, we have
\begin{eqnarray*}
\contprob{n}{\rho}{\sigma} & \leqslant & \varepsilon +
\contprob{n}{\rho}{\tau}
\end{eqnarray*}
by induction on $n$, using \cref{prop:decsurvive}.  For the
inductive step we reason
\begin{flalign*}
  \contprob{n+1}{\rho}{\sigma} =&\; \contprob{n}{\rho}{\sigma} - \sum_{s \in
                         \rhofailof{n} } \sigmapass(s) \\
                      \leqslant&\; \varepsilon +
                                 \contprob{n}{\rho}{\tau} -
                                 \sum_{s \in
                         \rhofailof{n}} \taupass(s)
  & \text{ (IH and $\tau \leqslant \sigma$)} \\
    =&\; \varepsilon + \contprob{n+1}{\rho}{\tau}
\end{flalign*}
 Take $n$ large enough that
$\contprob{n}{\rho}{\tau} \leqslant \varepsilon$.  Then  $\continf{\rho}{\sigma} \
\leqslant \ \contprob{n}{\rho}{\sigma} \ \leqslant \ 2\varepsilon$ as required.

Lastly we prove (\ref{item:vict}) $\Rightarrow$ (\ref{item:supff}),

 For each active-ending play $s$, define $\sigma.s$ to be the play-measure sending passive-ending $t$ to $\sigma(s.t)$.  Define $q(s)$ to be the
     supremum of all possible values of $w(\tau)$ for a finitely founded
     play-measure $\tau \leqslant \sigma.s$.

  For each passive-ending play $s.k$ we proceed similarly.  For $w
  \smin [0,1]$, say  that a \emph{$k$-play-measure} of weight $w$ is a family $(\sigma_i)_{i
   \in\arity{k}}$ of play-measures of weight $w$.  Define
  $\sigma.s.k$ to be the $k$-play-measure $(\sigma.s.k.i)_{i \in
    \arity{k}}$.  Define $q(s.k)$ to be the supremum of all
  possible values of $w(\tau)$ for a finitely founded $k$-play-measure 
  $\tau \leqslant \sigma.s.k$.  

  Thus we have $0 \leqslant q(s) \leqslant \sigma(s)$ for every active- or
  passive-ending $s$.  Furthermore, for every active-ending $s$, we
  have $q(s) = \sum_{k \in K} q(s.k)$, and for every passive-ending
  $s.k$, we have $q(s.k) = \inf_{i\in\arity{k}} q(s.k,i)$, taking this
  infimum to be $\sigma(s.k)$ if $\arity{k} = \emptyset$.  

Now, for all active- or passive-ending  $s$ we put $r(s) \eqdef
\sigma(s) - q(s)$.  Thus $0 \leqslant r(s) \leqslant \sigma(s)$.  Furthermore, for every active-ending
$s$, we have $r(s) = \sum_{k \in K} r(s,k)$, and for every
passive-ending $s.k$ we have $r(s.k) = \sup_{i\in\arity{k}}
r(s.k.i)$, taking this supremum to be 0 if $\arity{k} = \emptyset$.

We shall show that $r(0) = 0$ and hence $q(0) = w(\sigma)$, which
gives our required result.  Parenthetically,
it follows by induction that, for all active- or passive-ending $s$,
we have $r(s)
= 0$ and hence $q(s) =
\sigma(s)$.

For all $n \smin \nats$, put $a_n \eqdef 2 - \frac{1}{n+1}$.  This gives 
a sequence
\begin{math}
  1 = a_0 < a_1 < \cdots  
\end{math}.   
of numbers less than 2.  Choose a partial counterstrategy $\rho$ such that,
for any passive-ending $t.k$ with $|t| =n$ and $\arity{k} \not= \emptyset$, we have $r(\sigma(t.k)) \geqslant
\frac{a_{n}}{a_{n+1}} r(t.k)$ or equivalently $r(t.k) \leqslant
\frac{a_{n+1}}{a_{n}}r(\sigma(t.k))$.     For $t \smin \rho$ of length $n$, 
\begin{align*}
  r(t) =&\; \sum_{k \in K} r(t.k) \\*
  =&\; \sum_{\substack{k \in K \\ \arity{k} = \emptyset}} r(t.k) +
  \sum_{\substack{k \in K \\ \arity{k} \not= \emptyset}} r(t.k) \\
  \leqslant&\; 0 + \frac{a_{n+1}}{a_n} \sum_{\substack{k \in K \\ \arity{k} \not= \emptyset}} r(\sigma(t.k)) 
\end{align*}
Adding these equations over all such $t$ we get
\begin{align*}
    \sum_{\substack{t\in\rho \\ |t| = n}} r(t) \leqslant&\; 
                                                           \frac{a_{n+1}}{a_n}
                                                           \sum_{\substack{t
                                                          \in\rho \\
  |t| = n+1}} r(t)  
\end{align*}
and therefore
\begin{flalign*}
&& r(0) \leqslant &\; a_{n} \sum_{\substack{t\in\rho \\ |t| = n}}
                    r(t)  & \text{ (by induction on $n$)} \\
&&  \leqslant&\; 2 \sum_{\substack{t\in\rho \\ |t| = n}} \sigma(t)
 &\text{ (since $a_n
  < 2$ and $r(t) \leqslant \sigma(t)$)} \\
&& =&\; 2 \contprob{n}{\rho}{\sigma}
\end{flalign*}%
Since this holds for all $n \smin \nats$, we have
$r(0) \ \leqslant \ 2\continf{\rho}{\sigma} = 0$
  as required. 
   \end{proof}
  
   \begin{theorem}\label{thm:def}
     A strategy $\sigma$ is definable iff it is victorious.
   \end{theorem}
   \begin{proof}
        ($\Rightarrow$) follows from \cref{prop:presvictor}.
     To prove
     ($\Leftarrow$), we define for each $n \smin \nats$ a victorious
     strategy $\sigma_n$ of weight $2^{-n}$ and a finitely founded strategy $\tau_n$ of
     weight $2^{-n-1}$, as
     follows:
     \begin{itemize}
     \item We put $\sigma_0 \eqdef \sigma$.
     \item $\tau_n$ is a finitely founded strategy smaller or equal $\sigma_n$ of weight $2^{-n}$, which exists by
       \cref{prop:victoreq}(\ref{item:supff}).
     \item  We put $\sigma_{n+1} \eqdef \sigma_n - \tau_n$.
     \end{itemize}     
For all $n \smin \nats$ we have $\sigma = \sum_{k < n} \tau_k  +
\sigma_n$.  Since $\sigma_n (s) \leqslant w(\sigma_n) =
2^{-n-1}$ for all $s$, we have $\sigma = \sum_{k\in\nats}
\tau_k$.  \cref{prop:wfstratdef}(\ref{item:wfstratdef}) tells us that $\tau_k$ is the denotation
of some partial term $M_k$, so the program $\sum_{k\in\nats} M_k$
denotes $\sigma$.
   \end{proof}

   \section{Monads and theories}\label{sec:monadsth}
\subsection{Strategy monads} \label{sec:monads}
   
We now turn to the connections with monads on $\set$. Recall that $\dist$
and $\distfin$ are the monads of discrete and finitely supported distributions 
correspondingly. For the rest, let $S = (\arity{k})_{k\in\Ops}$ be a signature. 
  For any set $A$, write $S+A$ for the signature extending $S$ with 
  $A$ additional constants.  Then an \emph{$S$-term on $A$} is a term
  built up using operations in $S$ and $A$-many variables; this
  corresponds to a deterministic $(S+A)$-program.

  We define the following monads.
  \begin{definition}
    \begin{enumerate}
    \item The \emph{free monad} $T_S$ sends a set $A$ to the set of
      all $S$-terms on $A$.
     \item The monad ${\distfin}_{S}$ sends a set $A$ to the set of
       finitely probabilistic $(S+A)$-terms modulo bisimilarity.
     \item The monad $\dist_{S}$ sends a set $A$ to the set of
       countably probabilistic  $(S+A)$-terms modulo
       bisimilarity.
     \item The monad $\finfound_{S}$ sends a set $A$ to the set of all finitely founded $(S+A)$ strategies,
    corresponding to finitely probabilistic $(S+A)$-terms modulo trace
    equivalence.
   \item The monad $\vict_{S}$ sends a set $A$ to the set of all victorious $(S+A)$ strategies,
    corresponding to countably probabilistic $(S+A)$-terms modulo trace
    equivalence.
    \end{enumerate}
  \end{definition}

\subsection{Theories}

Moggi's work on representing computational effects via monads~\cite{Moggi91a} was
further developed by later authors~\cite{PlotkinPower02}, who presented many of
these monads via a theory.  We briefly recall the basics.

For a signature~$S$, an \emph{$S$-equation}  
is a triple $\tuple{B,s,t}$, where~$B$ is a set and
$t,t'$ are $S$-terms on $B$. This can be written as
\begin{displaymath}
(x_b)_{b \in B} \vdash s = t
\end{displaymath}
For example, the \emph{idempotence} equation for an
operation symbol~$f$ takes the form
\begin{displaymath}
  x \vdash f(x)_{i \in \arity{f}} = x 
\end{displaymath}
and the \emph{commutativity} equation between operation symbols~$f$ and $g$ takes the form
\begin{align*}
   (x_{i,j})_{i \in \arity{f}, j \in \arity{g}} \vdash&\\
      f(g(x_{i,j})_{j \in \arity{g}}&)_{i \in \arity{f}} =   g(f(x_{i,j})_{i \in \arity{f}})_{j \in \arity{g}}.
\end{align*}
A \emph{theory} consists of a signature
$S$ and a set $\cate$ of $S$-equations.    It is \emph{idempotent}
when it proves the idempotence of each operation symbol.\footnote{This corresponds to the notion of ``affineness'' for monads in~\cite{Kock71}.}    
Such a theory, if consistent, has no constants.

Each theory $\theta = (S,\cate)$ gives rise to 
a monad $T_{\theta}$, sending a set $A$ to the set of all
\emph{$\theta$-classes on $A$}, i.e.\ 
$S$-terms on $A$ modulo the least congruence that contains every 
instance of an equation in $\cate$.%

Here are some examples of theories:
\begin{itemize}
\item The \emph{convex theory} has  a binary
  operation $+_p$ for all $p \smin [0,1]$, and equations as in
  \cref{fig:prob-laws}(a).  It is idempotent, and yields the monad $\distfin$.
 \item The \emph{$\omega$-convex theory} has an
   $\omega$-ary operation $\sigma_p$ for each sequence $p = (p_n)_{n
     \in \nats}$ of nonnegative reals summing to~1, and equations as
   in \cref{fig:prob-laws}(b).  It is idempotent, and yields the monad $\dist$.
 \item The \emph{free theory} on a signature $S$ consists of $S$ and
   no equations.  This yields the free monad $T_S$.
 \item The \emph{sum} of theories $\theta$ and $\theta$ consists of
   all the operation symbols\footnote{Strictly speaking, we take the
     sum of the two signatures.}  and equations in $\theta$ and $\theta'$. The
   \emph{tensor} is defined the same way, by adding commutativity equations
   between all $\theta$-operations and $\theta'$-operations.
     The resulting monads are respectively called the \emph{sum} and \emph{tensor} of $T_{\theta}$ and $T_{\theta'}$.
\item A special case of the previous construction: the \emph{free extension} of a theory $\theta$ by a
  signature $S$, written $\theta +S$, consists of all the operation symbols in $\theta$ and
  $S$ and all equations in $\theta$.   The \emph{free
    tensorial extension}, written $\theta \tensor S$,  is defined the same way, with additional equations
  saying that each $\theta$-operation commutes with each $S$-operation.
 The resulting monads are called the \emph{free extension} and
  \emph{free tensorial extension} of $T_{\theta}$ by~$S$ respectively. 
\end{itemize}

\subsection{Characterizing strategy monads}

In order to relate strategy monads to theories, we first see how to
treat 
variables appearing in terms as extra constants:
\begin{proposition} \label{prop:var}
  Let $\theta$ be a theory and $S$ a signature.
  \begin{enumerate}
  \item \label{item:varsum} The $(\theta +S)$-classes on a set $A$ 
    correspond to the $(\theta + (S+A))$-classes on
    $\emptyset$.
 \item \label{item:vartensor} If $\theta$ is idempotent (e.g.\ is the convex or
   $\omega$-convex theory), then the $(\theta \tensor S)$-classes on a
   set $A$ 
    correspond to the $(\theta \tensor (S+A))$-classes on
    $\emptyset$.
  \end{enumerate}
\end{proposition}
\begin{proof}
  \begin{enumerate}
  \item Obvious.
  \item Because an idempotent operation commutes with any constant.\qed
  \end{enumerate}\noqed
\end{proof}
Now we have the following results, for any signature $S$.
\begin{theorem}
  \begin{enumerate}
  \item \label{item:bisfin} The monad ${\distfin}_{S}$ is the free extension of~$\distfin$ by $S$.
  \item \label{item:biscount} The monad $\dist_{S}$ is the free extension of~$\dist$ by $S$.
  \item \label{item:tracefin} The monad $\finfound_S$  is the free
    tensorial extension of~$\distfin$ by $S$.
  \item \label{item:tracecount} If $S$ is finitary, then the monad $\vict_{S}$ 
    is the free tensorial extension of~$\dist$ by $S$.
  \end{enumerate}
\end{theorem}
\begin{proof}
  We consider each monad applied to a set $A$.
  \begin{enumerate}
  \item The case $A = \emptyset$ is given by~\cref{pro:bisim}, and the
    general case follows via~\cref{prop:var}(\ref{item:varsum}).
  \item Likewise.
   \item The case $A = \emptyset$ is given by~\cref{prop:finstuff}(\ref{item:fincomplete}), and the
     general case follows via~\cref{prop:var}(\ref{item:vartensor}).
   \item The case $A = \emptyset$ is given by~\cref{thm:comp_fin}, and the
     general case follows via~\cref{prop:var}(\ref{item:vartensor}).\qed
  \end{enumerate}\noqed
\end{proof}

\begin{theorem}
  Tensor completeness is equivalent to the following statement: for
  every signature $S$, the monad $\vict_{S}$ is  the free tensorial extension of~$\dist$ by $S$.
\end{theorem}
\begin{proof}
  ($\Rightarrow$) follows from~\cref{prop:var}(\ref{item:vartensor}).
  For ($\Leftarrow$), apply the monad to $\emptyset$.
\end{proof}
   
\section{Conclusions}\label{sec:conc}
   
We have tacked the problem of completeness of (infinitary) tensor logic of 
programs with probabilistic choice, w.r.t.\ trace equivalence as the underlying 
semantic notion of program equivalence. In contrast to the analogous case of 
nondeterminism, the problem turned out to be remarkably difficult.
We managed to establish completeness for finitary signatures (\Cref{thm:comp_fin}), by developing
a system of normalization procedures, while also demonstrating that the present 
approach cannot scale to the infinitary case. The latter thus presents a hard 
open problem, which we aim to attack in further work. 

We have shown that the general completeness problem simplifies significantly 
and resolves into positive in presence of certain natural quasiequational laws, most obviously, 
in presence of the cancellativity law for the probabilistic choice. We refine this
furthermore by proving completeness \emph{up to impersonation} (\Cref{thm:upto-imp}), which
shows the necessary subtlety of any putative counterexample to completeness.

As noted in~\Cref{sec:t-sem}, the $\omega$-convex space of probabilistic strategies is cancellative.
But what about the $\omega$-convex space of programs modulo tensor equivalence?
By~\Cref{cor:up-to-cancel}, it is cancellative iff completeness holds, which is true for
finitary signatures but is unknown for infinitary ones. 

In many fields of semantics, an important problem, closely related to completeness, is the problem of definability,
which is the problem of describing intrinsically the space of possible semantics
of programs. We provide such a description (\Cref{thm:def}) through the class of \emph{victorious strategies},
which precisely capture semantics of probabilistic programs in game-theoretic terms.

Besides the problem of tensor completeness for infinitary signatures,
 future 
 work may consider continuous probability and the combination of probability with nondeterminism.

\bibliographystyle{IEEEtranDOI}
\bibliography{monads}

\end{document}